\documentclass[a4paper,twocolumn,11pt,accepted=2024-11-06]{quantumarticle}
\pdfoutput=1

\usepackage[T1]{fontenc}
\usepackage[latin9]{inputenc}
\usepackage{textcomp}
\usepackage{amssymb}
\usepackage{amsthm}
\usepackage{enumerate}
\usepackage{graphicx}
\usepackage{mathtools}
\usepackage{esint}
\usepackage{tikz}
\usepackage{subfigure}
\usepackage{hyperref}
\usepackage{bbold}
\usepackage[english]{babel}

\usepackage{bm}
\usepackage{braket}
\usepackage[numbers,sort&compress]{natbib}
\makeatletter

\usepackage{amsmath}

\newtheorem{defi}{Definition}
\newtheorem{theo}{Theorem}
\newtheorem{lem}{Lemma}
\newtheorem{algo}{Algorithm}

\newcommand{\ketbra}[1]{\ket{#1}\hspace{-0.3em}\bra{#1}}
\newcommand{\Tr}{\mathrm{Tr}}
\newcommand{\dm}[1]{d^{\,2}{#1}\,}
\newcommand{\bdm}[1]{\text{d}^{2m}\!\bm{#1}\,}
\newcommand{\bdn}[1]{\text{d}^{2n}\!\bm{#1}\,}

\newcommand{\PQD}{\hbox{PQD}}
\newcommand{\sPQD}{\hbox{($s$)-PQD}}
\newcommand{\SPQD}{\hbox{($\bm{s}$)-PQD}}

\newcommand{\negtPQD}{\hbox{($-t$)-PQD}}
\newcommand{\TPQD}{\hbox{($\bm{t}$)-PQD}}

\newcommand{\rhoin}{\rho_{\rm{in}}}

\@ifundefined{textcolor}{}
{%
 \definecolor{BLACK}{gray}{0}
 \definecolor{WHITE}{gray}{1}
 \definecolor{RED}{rgb}{1,0,0}
 \definecolor{GREEN}{rgb}{0,1,0}
 \definecolor{BLUE}{rgb}{0,0,1}
 \definecolor{CYAN}{cmyk}{1,0,0,0}
 \definecolor{MAGENTA}{cmyk}{0,1,0,0}
 \definecolor{YELLOW}{cmyk}{0,0,1,0}
 }
\makeatother

\begin{document} 

\title{Phase-space negativity as a computational resource for quantum kernel methods}


\author{Ulysse Chabaud}
\thanks{ulysse.chabaud@inria.fr}
\affiliation{DIENS, \'Ecole Normale Sup\'erieure, PSL University, CNRS, INRIA, 45 rue d'Ulm, Paris 75005, France}

\author{Roohollah Ghobadi}
\thanks{farid.ghobadi80@gmail.com}
\affiliation{Institute for Quantum Science and Technology, University of Calgary, Calgary, AB, T2N 1N4, Canada}

\author{Salman Beigi}
\thanks{salman.beigi@gmail.com}
\affiliation{School of Mathematics, Institute for Research in Fundamental Sciences (IPM), P.O. Box 19395-5746, Tehran, Iran}

\author{Saleh Rahimi-Keshari}
\thanks{s.rahimik@gmail.com}
\affiliation{School of Physics, Institute for Research in Fundamental Sciences (IPM), P.O. Box 19395-5531, Tehran, Iran}

\begin{abstract}
Quantum kernel methods are a proposal for achieving quantum computational advantage in machine learning. They are based on a hybrid classical-quantum computation where a function called the quantum kernel is estimated by a quantum device while the rest of computation is performed classically. Quantum advantages may be achieved through this method only if the quantum kernel function cannot be estimated efficiently on a classical computer. In this paper, we provide sufficient conditions for the efficient classical estimation of quantum kernel functions for bosonic systems. These conditions are based on phase-space properties of data-encoding quantum states associated with the quantum kernels: negative volume, non-classical depth, and excess range, which are shown to be three signatures of phase-space negativity.
We consider quantum optical examples involving linear-optical networks with and without adaptive non-Gaussian measurements, and investigate the effects of loss on the efficiency of the classical simulation. Our results underpin the role of the negativity in phase-space quasi-probability distributions as an essential resource in quantum machine learning based on kernel methods. 
\end{abstract}

\maketitle

\section{Introduction}

Small-scale quantum devices with a few hundred qubits~\cite{preskill,arute2019quantum} represent a novel paradigm for applications in quantum simulation~\cite{RMP14}, quantum chemistry~\cite{qchem}, and quantum machine learning~\cite{biamonte2017quantum,Farhi}. Despite their relatively small scale, there are strong evidences based on complexity-theoretic arguments, that simulating the sampling statistics of these devices is beyond the reach of classical computers~\cite{terhal2002adaptive, bremner2016average,aaronson2011computational,Science}. However, determining why quantum devices can outperform classical ones remains a fundamental question in quantum information science. In particular, identifying the necessary quantum resources and understanding the effects of errors on them is a complex challenge.

One approach to address this challenge, especially for quantum sampling problems, has been based on the use of quasi-probability distributions~\cite{CahGla69,Spekkens, mari2012positive,veitch2013efficient,rahimi2016sufficient,Stahlke,pashayan2015estimating}. In this approach, a classical description of a quantum experiment is possible as long as one can find probabilistic descriptions for the input state, the evolution, and the measurement in terms of non-negative quasi-probability distributions. Also, the relation between the negative volume of quasi-probability distributions and the computational overhead in estimating an output probability of a quantum circuit using those distributions has been investigated~\cite{Stahlke,pashayan2015estimating}. 


In this paper, we build upon these results to quantify the sampling complexity in estimating \textit{quantum kernels} in terms of the negativity of the associated quasi-probability distributions. Quantum kernel methods have been proposed as a promising near-term application of quantum computers \cite{mengoni2019kernel}, and the non-classicality of quantum kernels has been previously investigated in \cite{ghobadi2021nonclassical}, based on phase-space inequalities from \cite{bohmann2020phase}. Hereafter, we go beyond these conceptual characterizations and present sufficient conditions for efficient classical estimation of kernel functions in continuous-variable quantum machine learning. Informally, we show that one can achieve the same performance as a quantum circuit by using classical Monte Carlo-type sampling techniques if the negativity of the phase-space quasi-probability distributions (PQD) of data-encoding quantum states is limited. This condition is independent of the measurement strategy and holds even if sampling from the output probability distribution of the circuit is classically hard, such as in boson sampling problems.

In light of this, we identify the negativity of \PQD s as a necessary resource to achieve quantum computational speedups in machine learning based on kernel methods. More precisely, we obtain classical algorithms with fine-grained sample complexity based on specific phase-space properties, namely negative volume~\cite{kenfack2004negativity,albarelli2018resource}, non-classical depth~\cite{lee1991measure,Sabapathy2016multimode} and excess ranges. While the negative volume is a direct measure of the negativity of a \PQD, the non-classical depth indicates the amount of thermal noise necessary to render a \PQD\ non-negative, and we further show that excess range of \PQD s can also be understood a signature of phase-space negativity. As illustrating examples, we discuss efficient classical estimation of quantum kernel functions that are based on Gaussian states, non-Gaussian output states of linear optical networks, and partially measured Gaussian states, including adaptive measurement strategies.

Our results involve a new approach to classical estimation of quantum kernel functions beyond probability estimation~\cite{Stahlke,pashayan2015estimating} that takes into account the structure of state preparation. A surprising consequence is that estimating quantum kernel functions based on a large class of partially measured Gaussian states can be done efficiently by classical computers despite these states being highly non-Gaussian (see Theorem~\ref{th:part-meas-G}).

The structure of the paper is as follows: in Section~\ref{sec:QML}, we set the stage and provide some background on quantum kernel methods; in Section~\ref{sec:kern-phase-space}, we derive general expressions for quantum kernel functions for bosonic systems based on \PQD s (see Lemma~\ref{lem:kern-expr}); in Section~\ref{sec:est}, we introduce classical Monte Carlo sampling methods and derive two phase-space-inspired classical algorithms for estimating quantum kernel functions (see Algorithms~\ref{algo:MCQD-kern1} and~\ref{algo:MCQD-kern2}), together with rigorous performance guarantees; we illustrate the flexibility and applicability of our algorithms in Section~\ref{sec:examples}, by applying them to several examples of high significance in bosonic quantum information processing, identifying the relevant quantum computational resources in each setting; we conclude in Section~\ref{sec:concl}.

\section{Quantum Kernel Methods}
\label{sec:QML}


In quantum machine learning, we often model an unknown function $f$ of some classical data $x$ by $f(x)= \Tr[\rho(x)O]$ where $\rho(x)$ is a quantum state depending on the input $x$, and $O$ is a quantum observable. 

A general encoding of classical data $x$ to a quantum state $\rho(x)$ over $m$ subsystems based on a quantum map may take the form
\begin{equation}\label{eq:norm-state}
x\mapsto\rho(x)=\frac{\mathrm{Tr}_k[(\Pi(x)\otimes\mathbb I_m)U(x)\rhoin(x)U(x)^\dag]}{\mathrm{Tr}[(\Pi(x)\otimes\mathbb I_m)U(x)\rhoin(x)U(x)^\dag]},
\end{equation}
where the classical data $x$ may be encoded in a density operator $\rhoin(x)$ describing an initial quantum state, a unitary operator $U(x)$ describing a quantum circuit, and a positive operator-valued measure (POVM) element $\Pi(x)$ describing the measurement, and where the subscript $k$ denotes that the first $k$ subsystems are being measured. The denominator is for normalisation of the post-measurement state
\begin{equation}\label{eq:post-meas-state}
\rho_{\Pi(x)}:=\mathrm{Tr}_k[(\Pi(x)\otimes\mathbb I_m)U(x)\rhoin(x)U(x)^\dag].
\end{equation}


It is often the case that classical data is encoded in the input state or the circuit parameters only, since these are the parameters that can be easily varied in practice~\cite{schuld2021supervised}. Note also that the dependency of $\rhoin$ or $\Pi$ on the classical data $x$ can be moved to $U(x)$ without loss of generality.  Similarly, the initial state $\rhoin$ can be chosen as a tensor product state without loss of generality by changing  $U$. That being said, we allow for very general encoding strategies and let the input, evolution and measurement potentially depend on classical data $x$.

In this setting, the encoding map $x\mapsto[\rhoin(x),U(x),\Pi(x)]$ is known, while the observable $O$ is unknown. The goal is to approximate $f(x)$ given a training dataset $\mathcal{D}=\{(x_{i},y_{i}=f(x_i))\}_{i=1}^{n}$. To this end, we look for an observable $\widetilde O$ that minimizes the cost function
\begin{equation}\label{eq:model}
\frac{1}{n} \sum_{i=1}^n c\big(\Tr\big(\rho(x_i) \widetilde O\big) , y_i  \Big) + \lambda\Tr\big(\widetilde O^2\big).
\end{equation}
Here, $c(\cdot, \cdot)\geq 0$ is a cost function which measures the distance of the predicted value $\Tr\big(\rho(x_i) \widetilde O\big)$ and the true value $y_i$. The second term with the regularization parameter
$\lambda>0$ ensures avoiding overfitting. This regularization term penalizes complex hypotheses that best match the training dataset but do not provide a good prediction for arbitrary inputs. 

Invoking the representer theorem, it can be shown that the optimal observable $\widetilde O$ that minimizes \eqref{eq:model} can be written as~\cite{scholkopf2001generalized, hofmann2008kernel}  
\begin{equation}
\widetilde O=\sum_{i=1}^{n}\alpha_{i}\rho(x_{i}),
\end{equation}
where $\alpha_i$'s are real numbers. In this case, the optimal approximating function equals 
\begin{equation}\label{eq:fopt}
\widetilde f(x)=\sum_{i=1}^{n}\alpha_{i}K(x,x_{i}),
\end{equation}
where we have introduced the \emph{kernel function} 
\begin{equation}\label{eq:kernel}
K(x,x')=\Tr\big[\rho(x)\rho(x')\big]=\frac{\Tr[\rho_{\Pi(x)}\rho_{\Pi(x')}]}{\Tr[\rho_{\Pi(x)}]\Tr[\rho_{\Pi(x')}]},
\end{equation}
where the right hand side is obtained using Eqs.~(\ref{eq:norm-state}) and (\ref{eq:post-meas-state}). This kernel function can be viewed as a measure of similarity between data points. 

According to Eq.~\eqref{eq:fopt}, in order to find the optimal function $\smash{\widetilde f(x)}$ we only need to compute the values of the kernel function; given two data points $x, x'$, we need to be able to compute the overlap of $\rho(x), \rho(x')$. The quantum kernel methods~\cite{havlivcek2019supervised,schuld2019quantum} are hybrid classical-quantum machine learning techniques which involve estimating these overlaps quantumly to within an inverse-polynomial additive error in the size of the corresponding quantum states, while the rest of the computation, i.e., computing the coefficients $\alpha_i$'s from the kernel values $K(x, x_i)$ is done classically. In particular, the quantum part of the computation consists in generating a polynomial number of copies of the states $\rho(x), \rho(x')$ and using these copies to estimate their overlap up to inverse-polynomial precision, for instance using the SWAP test~\cite{buhrman2001quantum}, which effectively implements a projection $\Pi_\text{Sym}$ onto the symmetric subspace (see Fig.~\ref{fig:Q-kern-est}). In the case of an encoding involving quantum measurements, this is only efficient when the probability to generate the states $\rho(x), \rho(x')$ is at least inverse-polynomially large, i.e., $\mathrm{Tr}[\rho_{\Pi(x)}]\ge1/\mathrm{poly}(m),\mathrm{Tr}[\rho_{\Pi(x')}]\ge1/\mathrm{poly}(m)$. When this is the case we say that the encoding is \textit{quantum-efficient}. As such, quantum computational advantage using quantum kernel methods may only come from the estimation of kernel values in the case of quantum-efficient encodings.

In the rest of this paper, we show using phase-space techniques that for a large family of quantum-efficient encoding schemes $x\mapsto\rho(x)$, the kernel values~\eqref{eq:kernel} can be estimated efficiently classically up to the same precision as can be done quantumly. When that is the case, there can be no computational advantage from quantum kernel methods.

\begin{figure}[t]
	\begin{center}
		\includegraphics[width=1\linewidth]{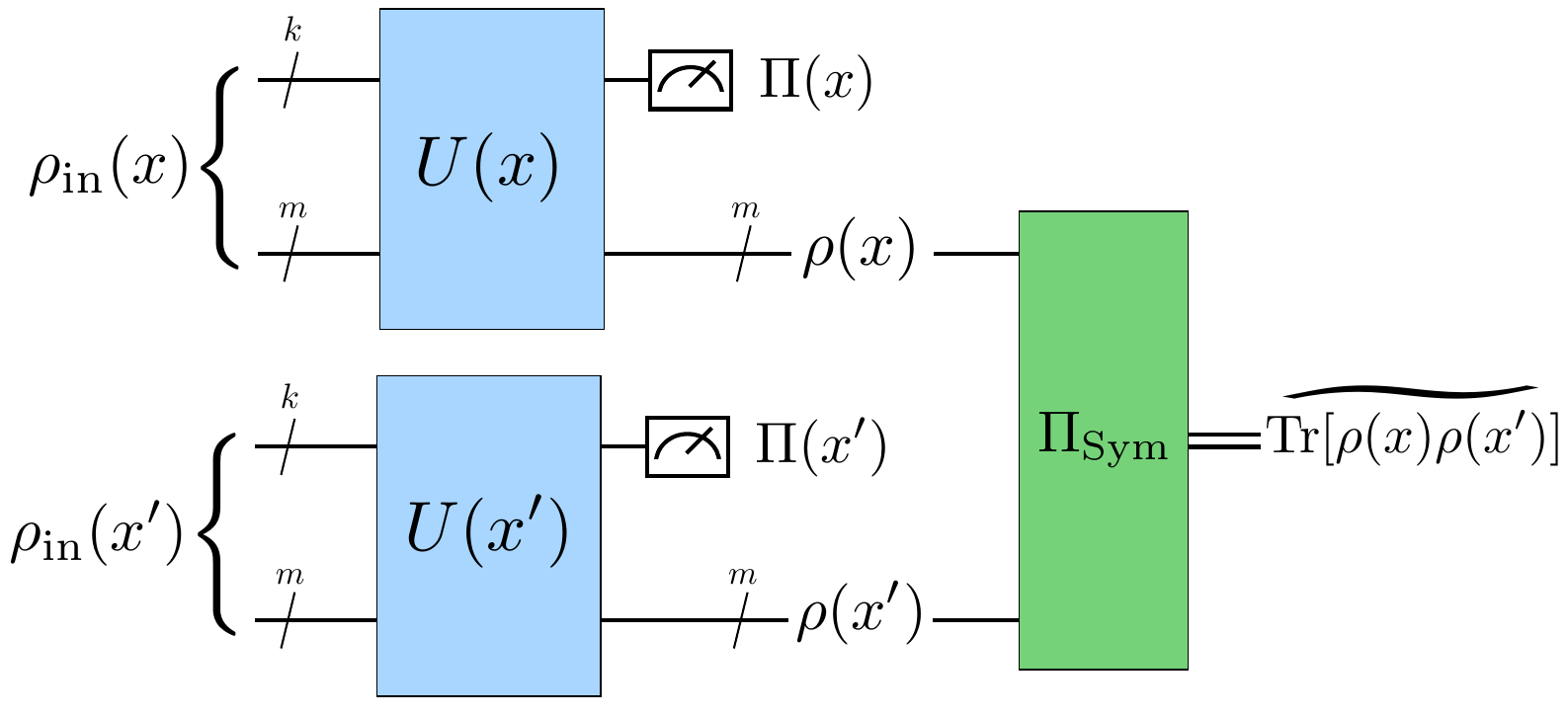}
		\caption{Quantum estimation of quantum state overlap. The green block $\Pi_\text{sym}$ represents the projection of $\rho(x)\otimes\rho(x')$ onto the symmetric subspace. Repeating that projection $N$ times allows to estimate the overlap $K(x,x')$ up to inverse-polynomial precision (in $N$) with exponentially small failure probability, by computing the frequency of successful projection.}
		\label{fig:Q-kern-est}
	\end{center}
\end{figure}

\section{Kernel functions in phase space}
\label{sec:kern-phase-space}

Let $\rho$ be a quantum state of an $m$-mode bosonic system. Then $\rho$ can be represented as~\cite{CahGla69,HilOCo84}:
\begin{equation}\label{eq:rho-PQD}
\rho=\pi^m\!\int_{\mathbb C^m}\! \bdm{\bm{\alpha}} W^{(\bm s)}_\rho(\bm{\alpha}) \Delta^{(-\bm s)}(\bm\alpha),
\end{equation}
where $\bm\alpha$ is a row vector of $m$ complex numbers, where $\bdm{\bm{\alpha}}=\mathrm d^m{\,\Re(\bm{\alpha})}\mathrm d^m{\,\Im(\bm{\alpha})}$, and where the $\bm s$-ordered phase-space quasi-probability distributions [\SPQD{s}] are defined by
\begin{equation}\label{eq:Wigner}
W^{(\bm s)}_\rho(\bm{\alpha})=\Tr[\rho \Delta^{(\bm s)}(\bm\alpha)],
\end{equation}
with phase-space point operators
\begin{equation}\label{eq:Delta-frame}
	\Delta^{(\bm s)}(\bm\alpha)=\int_{\mathbb C^m}\frac{\bdm{{\xi}}}{\pi^{2m}}\,D(\bm\xi)\, e^{\bm\xi\bm{s}\bm\xi^\dagger/2}
	e^{\bm\alpha\bm\xi^{\dagger}-\bm\xi\bm\alpha^{\dagger}}.
\end{equation}
Here, $\bm\alpha, \bm\xi$ are row vectors of $m$ complex numbers, the diagonal matrix $\bm s=\mathrm{diag}(s_1,s_2,\ldots,s_m)$ specifies the ordering for each mode, and $ D(\bm\xi)=\exp(\bm\xi\bm a^\dagger-\bm a\bm\xi^\dagger)$ is the $m$-mode displacement operator with $\bm{a}=( a_1,\ldots,a_m)$ and $\bm{a}^\dagger=( a_1^\dagger,\ldots,a_m^\dagger)^T$ being vectors of annihilation and creation operators, respectively. Note that these formulas are formally valid for more general ordering matrices, but we restrict to diagonal ones for simplicity. 

Using $\int\bdm{\alpha}\exp(\bm\alpha\bm\xi^{\dagger}-\bm\xi\bm\alpha^{\dagger})=\pi^{2m}\delta^{2m}(\bm\xi)$ and Eq.~\eqref{eq:Delta-frame}, one can verify that the \SPQD\ of density operators are normalized. In general, they can take negative values or do not represent probabilities of mutually exclusive events, which is why they are named quasi-probability distributions.

When $\bm s$ is the identity matrix $\bm I$, the \SPQD\ becomes the Glauber--Sudarshan representation which is a highly singular distribution for most quantum states. Only for classical states, which can be viewed as statistical mixtures of coherent states, the Glauber--Sudarshan $P$ function is non-negative. Other special cases are the Wigner function for $\bm s=\bm0$ that takes on negative values for some non-classical states, and the Husimi function for $\bm s=-\bm I$, in which case $W^{(-\bm I)}_\rho(\bm{\alpha})=\bra{\bm\alpha}\rho\ket{\bm\alpha}/\pi^m$, where $\ket{\bm\alpha}$ denotes an $m$-mode coherent state. We note that for the Husimi $Q$ function we always have $0\leq W^{(-\bm I)}_\rho(\bm{\alpha})\leq 1/\pi^m$.

These phase-space quasi-probability distributions allow us to obtain useful expressions for quantum kernel functions. Denoting by $\oplus$ the direct sum of matrices we have:

\begin{widetext}
\begin{lem}\label{lem:kern-expr}
The kernel function $K(x,x')=\Tr\left[\rho(x)\rho(x')\right]$ can be expressed as:
\begin{equation}\label{eq:kern-sPQD}
	K(x,x')=\pi^m\!\!\int_{\mathbb C^m}\!\mathrm d^{2m}{\bm{\gamma}} W^{(-\bm s)}_{\rho(x')}(\bm{\gamma})W^{(\bm s)}_{\rho(x)}\left(\bm{\gamma}\right),
\end{equation}
for all $\bm s=\mathrm{diag}(s_1,s_2,\ldots,s_m)$.
For a general encoding $x\mapsto\rho(x)=\rho_{\Pi(x)}/\mathrm{Tr}[\rho_{\Pi(x)}]$, it is given by $K(x,x')=\frac{\mathrm{Tr}[\rho_{\Pi(x)}\rho_{\Pi(x')}]}{\mathrm{Tr}[\rho_{\Pi(x)}]\mathrm{Tr}[\rho_{\Pi(x')}]}$, where
\begin{align}
	\label{eq:kern-sPQD-norm1}\mathrm{Tr}[\rho_{\Pi(x)}]&=\pi^k\!\!\int_{\bm\alpha\in\mathbb C^k}\!\!\!\!\!\!\mathrm d^{2k}\bm\alpha\,W_{\Pi(x)}^{(-\bm u)}(\bm\alpha)\int_{\bm\gamma\in\mathbb C^m}\!\!\!\!\!\!\!\mathrm d^{2m}\bm\gamma\,W_{U(x)\rhoin(x)U(x)^\dag}^{(\bm u\oplus\bm s)}(\bm\alpha,\bm\gamma),\\
	\label{eq:kern-sPQD-norm2}\mathrm{Tr}[\rho_{\Pi(x')}]&=\pi^k\!\!\int_{\bm\beta\in\mathbb C^k}\!\!\!\!\!\!\mathrm d^{2k}\bm\beta\,W_{\Pi(x')}^{(-\bm v)}(\bm\beta)\int_{\bm\gamma\in\mathbb C^m}\!\!\!\!\!\!\!\mathrm d^{2m}\bm\gamma\,W_{U(x')\rhoin(x')U(x')^\dag}^{(\bm v\oplus\bm t)}(\bm\beta,\bm\gamma),
\end{align}
for all $\bm s=\mathrm{diag}(s_1,s_2,\ldots,s_m)$, $\bm u=\mathrm{diag}(u_1,u_2,\ldots,u_k)$,  $\bm t=\mathrm{diag}(t_1,t_2,\ldots,t_m)$, and $\bm v=\mathrm{diag}(v_1,v_2,\ldots,v_k)$, and where
\begin{equation}\label{eq:kern-sPQDgen}
    \begin{aligned}
	   \mathrm{Tr}[\rho_{\Pi(x)}\rho_{\Pi(x')}]&=\pi^{m+2k}\!\!\int_{\bm\alpha,\bm\beta\in\mathbb C^k}\!\!\!\!\!\!\!\!\!\!\mathrm d^{2k}\bm\alpha \,\mathrm d^{2k}\bm\beta \,W_{\Pi(x)}^{(-\bm u)}(\bm\alpha)W_{\Pi(x')}^{(\bm v)}(\bm\beta)\\
    &\quad\quad\quad\times\int_{\bm\gamma\in\mathbb C^m}\!\!\!\!\!\!\!\mathrm d^{2m}\bm\gamma\,W_{U(x)\rhoin(x)U(x)^\dag}^{(\bm u\oplus\bm s)}(\bm\alpha,\bm\gamma)W_{U(x')\rhoin(x')U(x')^\dag}^{(-\bm v\oplus-\bm s)}(\bm\beta,\bm\gamma).
    \end{aligned}
\end{equation}
In the case of a unitary encoding $x\mapsto\rho(x)=U(x)\rhoin(x)U(x)^\dag$, this simplifies to:
\begin{equation}\label{eq:kern-sPQDunitary}
	K(x,x')=\pi^m\!\!\int_{\mathbb C^m}\! \mathrm d^{2m}{\bm{\gamma}} W^{(\bm s)}_{U(x)\rhoin(x)U(x)^\dag}(\bm{\gamma})W^{(-\bm s)}_{U(x')\rhoin(x')U(x')^\dag}(\bm{\gamma}).
\end{equation}
\end{lem}
\end{widetext}

\begin{proof}
The expression in Eq.~(\ref{eq:kern-sPQD}) is a direct consequence of Eq.~(\ref{eq:rho-PQD}). For the expressions in Eqs.~(\ref{eq:kern-sPQD-norm1}), (\ref{eq:kern-sPQD-norm2}) and (\ref{eq:kern-sPQDgen}), let $\sigma$ be a density operator over $k+m$ modes, $\Pi$ a POVM element over $k$ modes and $\rho:=\Tr_k[(\Pi\otimes I_m)\sigma]$. For all $\bm\gamma\in\mathbb C^m$, all $\bm s=\mathrm{diag}(s_1,\dots,s_m)$ and all $\bm u=\mathrm{diag}(u_1,\dots,u_k)$,
\begin{align}
	\nonumber W_\rho^{(\bm s)}(\bm\gamma)&=\Tr[\rho\Delta^{(\bm s)}(\bm\gamma)]\\
	\nonumber&=\Tr[\Tr_k[(\Pi\otimes I_m)\sigma]\Delta^{(\bm s)}(\bm\gamma)]\\
	\nonumber&=\Tr[\sigma(\Pi\otimes\Delta^{(\bm s)}(\bm\gamma))]\\
	&=\pi^k\!\!\int_{\bm\alpha\in\mathbb C^k}\!\!\!\!\!\mathrm d^{2k}\bm\alpha W_\Pi^{(-\bm u)}\!(\bm\alpha)\\
 \nonumber&\quad\quad\quad\times\Tr[\sigma(\Delta^{(\bm u)}\!(\bm\alpha)\!\otimes\!\Delta^{(\bm s)}\!(\bm\gamma))]\\
	\nonumber&=\pi^k\!\!\int_{\bm\alpha\in\mathbb C^k}\!\!\!\!\!\mathrm d^{2k}\bm\alpha W_\Pi^{(-\bm u)}(\bm\alpha)W_\sigma^{(\bm u\oplus\bm s)}(\bm\alpha,\bm\gamma),
\end{align}
where we used Eq.~(\ref{eq:rho-PQD}) in the fourth line and Eq.~(\ref{eq:Wigner}) in the first and last lines. With Eq.~(\ref{eq:post-meas-state}), integrating this relation for $\smash{W_{\rho_{\Pi(x)}}^{(\bm s)}}$ or $\smash{W_{\rho_{\Pi(x')}}^{(\bm t)}}$ yields Eq.~(\ref{eq:kern-sPQD-norm1}) or (\ref{eq:kern-sPQD-norm2}), respectively, while using this relation twice to expand $\smash{W_{\rho_{\Pi(x)}}^{(\bm s)}}$ and $\smash{W_{\rho_{\Pi(x')}}^{(-\bm s)}}$ yields Eq.~(\ref{eq:kern-sPQDgen}). Finally, the last expression in Eq.~(\ref{eq:kern-sPQDunitary}) is a direct consequence of Eq.~(\ref{eq:kern-sPQD}).
\end{proof}

\noindent Note that the above Lemma involves phase-space representations of POVM elements that are not necessarily trace-class operators. These should be understood as defined formally in the sense of distributions, i.e., by their inner product with classes of well-behaved functions, which is how they will be employed hereafter.

The rest of the paper is devoted to showing how the general expressions obtained in Lemma~\ref{lem:kern-expr} enable direct estimation of kernel functions using classical Monte Carlo methods, when various assumptions on the \SPQD s are involved. 

\section{Classical estimation of kernel functions}
\label{sec:est}

In this section, we describe two classical algorithms using Monte Carlo sampling methods to estimate kernel functions based on the expressions obtained in Lemma~\ref{lem:kern-expr}. The first Algorithm \ref{algo:MCQD-kern1} is obtained by treating kernel estimation as a probability estimation task, and following the approach of~\cite{pashayan2015estimating,Stahlke,Chakh2017} (see Fig.\ \ref{fig:algo1}). The second Algorithm \ref{algo:MCQD-kern2} is a non-trivial generalisation which takes advantage of the specifics of state preparation, and uses the first algorithm as a subroutine (see Fig.\ \ref{fig:algo2}).

\subsection{Monte Carlo estimation}
\label{sec:MCest}

Monte Carlo methods are standard algorithms for estimating classical expectation values: given a probability density function $y\mapsto P(y)$ and a bounded function $f$ over $\mathbb R^n$, the expectation value $\mathbb E_P[f]=\int_{\mathbb R^n}f(y)P(y)dy$ over the probability density $P$ can be estimated by computing the mean $\frac1N\sum_{j=1}^Nf(y_j)$ of the function $f$ over a finite number $N$ of samples $y_1,\dots,y_N$ from the probability density $P$.
The error associated with this estimation and the probability of failure are related through Hoeffding's inequality~\cite{Hoeffding}:
\begin{equation}\label{eq:Hoef}
	\text{Pr}\bigg[\bigg|\frac{1}{N}\sum_{j=1}^{N}f(y_j)-\mathbb E_P[f]\bigg|>\epsilon\bigg]\leq 2 \exp\!\bigg(\!{-\frac{2N\epsilon^{2}}{\mathcal R(f)^{2}}}\bigg),
\end{equation}
where $\mathcal R(f)= \max_{y\in\mathbb R^n}f(y)-\min_{y\in\mathbb R^n}f(y)$ is the range of the function $f$. Hence, if the number of samples satisfies $N\ge\frac1{2\epsilon^{2}}\mathcal R(f)^2\ln\!\left(\frac2\delta\right)$, then with probability at least $1-\delta$ the classical estimate of the expectation value $\mathbb E_P[f]$ has additive error less than $\epsilon$. 

In general, the kernel functions we wish to estimate are of the form $\Tr[\rho A]$, where $\rho$ and $A$ are bounded, positive operators (see Eq.~(\ref{eq:kern-sPQD})). In this case, the kernel function can be viewed as an expectation value over \PQD s rather than true probability distributions. A simple extension of the above method allows us to estimate such quantities: following the approach of~\cite{pashayan2015estimating,Stahlke,Chakh2017}, given the \TPQD , $\bm\mu\mapsto W^{(\bm{t})}_\rho(\bm\mu)$ over $\mathbb C^n$, representing an operator $\rho$, we define the probability distribution
\begin{equation}\label{eq:prob-dis}
	P(\bm{\mu}):=\frac{1}{\mathcal N\big(W^{(\bm{t})}_\rho\big)} \big\vert W^{(\bm{t})}_\rho(\bm\mu) \big\vert,
\end{equation}   
where	$\mathcal N(W^{(\bm{t})}_\rho):=\int_{\mathbb C^n}\bdn\mu\, \big|W^{(\bm{t})}_\rho(\bm\mu)\big|$ is the \textit{negative volume}, a measure of negativity in the \TPQD. We note that $\mathcal N(W^{(\bm{t})}_\rho)=1$ if the distribution is non-negative. Next, using the quasi-probability distribution $W^{(-\bm{t})}_A(\bm\mu)$ representing operator $A$, we introduce the estimator 
\begin{equation}\label{eq:estimator}
	E(\bm\mu):=\pi^n\mathcal N\big(W^{(\bm{t})}_\rho\big) \text{sgn}\big[W^{(\bm{t})}_\rho(\bm\mu)\big]\, W^{(-\bm{t})}_A(\bm\mu),
\end{equation} 
where $\text{sgn}[\cdot]=\pm1$ is the sign function.  By construction, the above quantity provides an unbiased estimator for $\Tr[\rho A]$ through Eq.~\eqref{eq:rho-PQD}:
\begin{equation}
	\mathbb E_P[E]=\int_{\mathbb C^n}\! \bdn{\bm{\mu}} P(\bm{\mu})\, E(\bm\mu)=\Tr[\rho A].
\end{equation}
Therefore, assuming that it is possible to evaluate $E$ (and in particular to compute $\mathcal{N}(W^{(\bm{t})}_\rho)$) and to generate samples from the probability distribution $P$ efficiently classically, $\Tr[\rho A]$ can be estimated using the following procedure: with input \TPQD s $W^{(\bm{t})}_\rho$ and $W^{(-\bm{t})}_A$ for a density operator $\rho$ and a bounded, positive operator $A$, respectively,
\begin{enumerate}[(i)]
	\item randomly sample $N$ outcomes $\bm\mu_1,\dots,\bm\mu_N$ from the probability density $P(\bm{\mu})$, defined by Eq.~\eqref{eq:prob-dis};
	\item using Eq.~\eqref{eq:estimator}, compute the corresponding values of the estimator $E(\bm\mu_1), \dots, E(\bm\mu_N)$;
	\item output the sample mean $\frac1N\sum_{j=1}^NE(\bm\mu_j)$.
\end{enumerate}

\noindent The error associated with this estimation and the probability of failure are once again related through Hoeffding's inequality \eqref{eq:Hoef}, for the estimator $E$. Hence, if the number of samples satisfies
\begin{equation}\label{eq:N-samples}
	N\ge\frac1{2\epsilon^{2}}\mathcal R(E)^2\ln\!\left(\frac2\delta\right)\!,
\end{equation}
then with probability at least $1-\delta$ the classical estimate of the expectation value has additive error less than $\epsilon$. Here, the range of the estimator is bounded as
\begin{equation}\label{eq:range-E}
\mathcal R(E)\le2\mathcal N\big(W^{(\bm{t})}_\rho\big) \mathcal R(W^{(-\bm{t})}_A),
\end{equation}
with $\mathcal R(W^{(-\bm{t})}_A)=\pi^n[\max_{\bm\mu}W^{(-\bm{t})}_A(\bm\mu)-\min_{\bm\mu}W^{(-\bm{t})}_A(\bm\mu)]$.
This implies that the complexity of the estimation procedure, determined by the number of samples to achieve a desired precision,  depends directly on the range of the function $W^{(-\bm{t})}_A$ and the negative volume of the \PQD\ $W^{(\bm{t})}_\rho$. Note that $\pi^n\max_{\bm\mu}|W^{(-\bm{t})}_A(\bm\mu)|\le\mathcal R(W^{(-\bm{t})}_A)\le2\pi^n\max_{\bm\mu}|W^{(-\bm{t})}_A(\bm\mu)|$, so we can alternatively use the extremal values rather than the range without affecting the scaling of the sample complexity.

Using this estimation procedure for $\Tr[\rho A]$ to estimate the overlaps in Lemma~\ref{lem:kern-expr} provides two classical algorithms for kernel estimation.

\begin{algo}\label{algo:MCQD-kern1}
Choosing $\rho=\rho(x)$ and $A=\rho(x')$, the above estimation procedure provides an additive estimate of the kernel $K(x,x')=\Tr[\rho(x)\rho(x')]$.  
\end{algo}

\noindent We provide a schematic depiction of this algorithm in Fig.~\ref{fig:algo1}. The correctness of Algorithm~\ref{algo:MCQD-kern1} is a direct consequence of Eq.~(\ref{eq:kern-sPQD}) in Lemma~\ref{lem:kern-expr} and its efficiency is given by Eqs.~(\ref{eq:N-samples}) and (\ref{eq:range-E}), which can be optimized by the choice of \SPQD.

Alternatively, we may estimate independently $\Tr[\rho_{\Pi(x)}]$, $\Tr[\rho_{\Pi(x')}]$ and $\Tr[\rho_{\Pi(x)}\rho_{\Pi(x')}]$:

\begin{algo}\label{algo:MCQD-kern2}
\begin{enumerate}[(i)]
    \item Choosing $\rho$ as the $k$-mode partial trace of $U(x)\rhoin(x)U(x)^\dag$ over the last $m$ modes and $A=\Pi(x)$, the above estimation procedure provides an additive estimate of $\Tr[\rho_{\Pi(x)}]$.
    \item Choosing $\rho$ as the $k$-mode partial trace of $U(x')\rhoin(x')U(x')^\dag$ over the last $m$ modes and $A=\Pi(x')$, the above estimation procedure provides an additive estimate of $\Tr[\rho_{\Pi(x')}]$.
    \item Choosing $\rho=\sigma(x,x')$ as the $(2k)$-mode overlap of the last $m$ modes of the $(k+m)$-mode states $U(x)\rhoin(x)U(x)^\dag$ and $U(x')\rhoin(x')U(x')^\dag$ and choosing $A=\Pi(x)\otimes\Pi(x')$, the above estimation procedure provides an additive estimate of $\Tr[\rho_{\Pi(x)}\rho_{\Pi(x')}]$.
    \item Computing the ratio of these three estimates provides an additive estimate of the kernel $K(x,x)=\Tr[\rho_{\Pi(x)}\rho_{\Pi(x')}]/(\Tr[\rho_{\Pi(x)}]\Tr[\rho_{\Pi(x')}])$.
\end{enumerate}
\end{algo}

\noindent We provide a schematic depiction of this algorithm in Fig.~\ref{fig:algo2}. The correctness of Algorithm~\ref{algo:MCQD-kern2} is a consequence Eqs.~(\ref{eq:kern-sPQD-norm1}), (\ref{eq:kern-sPQD-norm2}) in Lemma~\ref{lem:kern-expr} for the estimations of the first two terms and of Eq.~(\ref{eq:kern-sPQDgen}) in Lemma~\ref{lem:kern-expr} for the estimation of the third term. Once again, the efficiency of each step is given by Eqs.~(\ref{eq:N-samples}) and (\ref{eq:range-E}), and can be optimized by the choice of \SPQD. 

\noindent  

\begin{figure}[t]
	\begin{center}
		\includegraphics[width=1\linewidth]{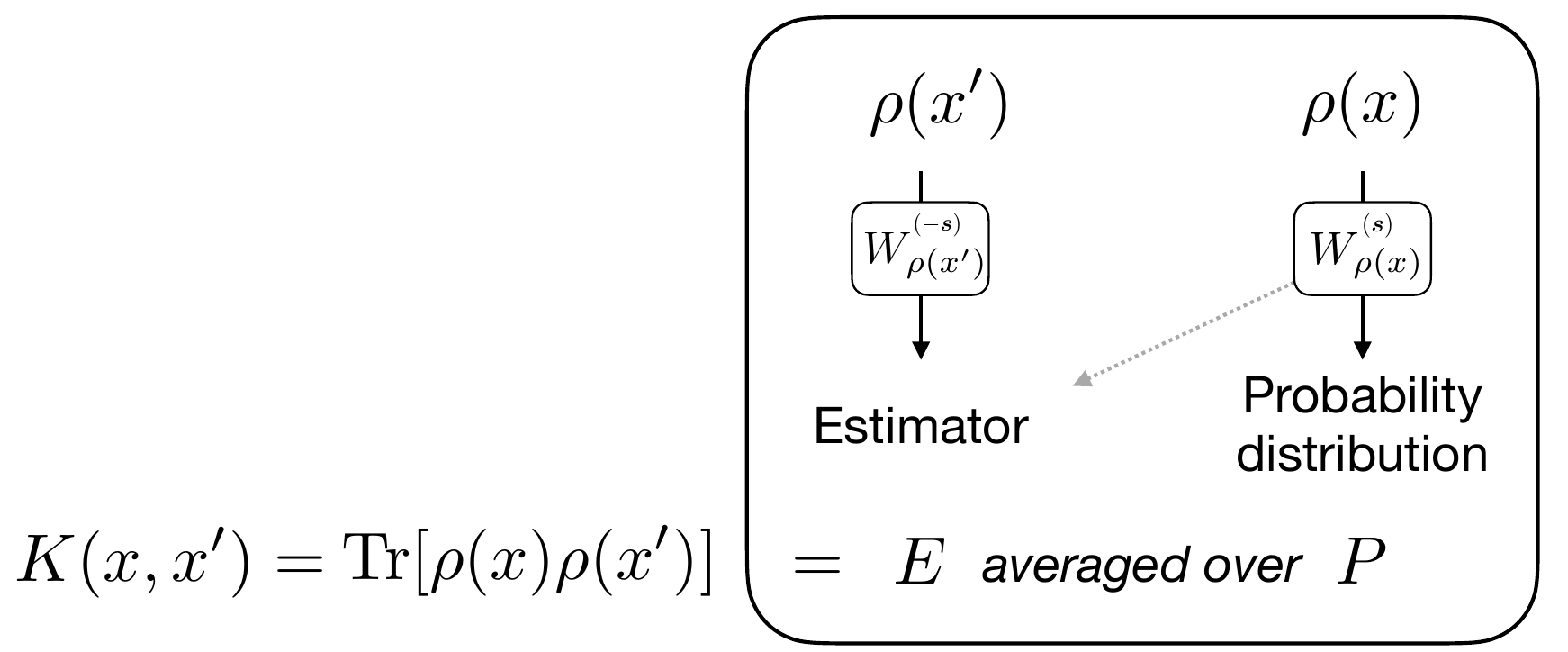}
		\caption{Classical estimation of quantum state overlap, following Algorithm~\ref{algo:MCQD-kern1}. The dashed arrow indicates that the negativity information of $W_{\rho(x')}^{(\bm s)}$ is also being used to define the estimator $E$ according to Eq.~(\ref{eq:estimator}).}
		\label{fig:algo1}
	\end{center}
\end{figure}

\begin{figure*}[!t]
	\begin{center}
		\includegraphics[width=0.9\linewidth]{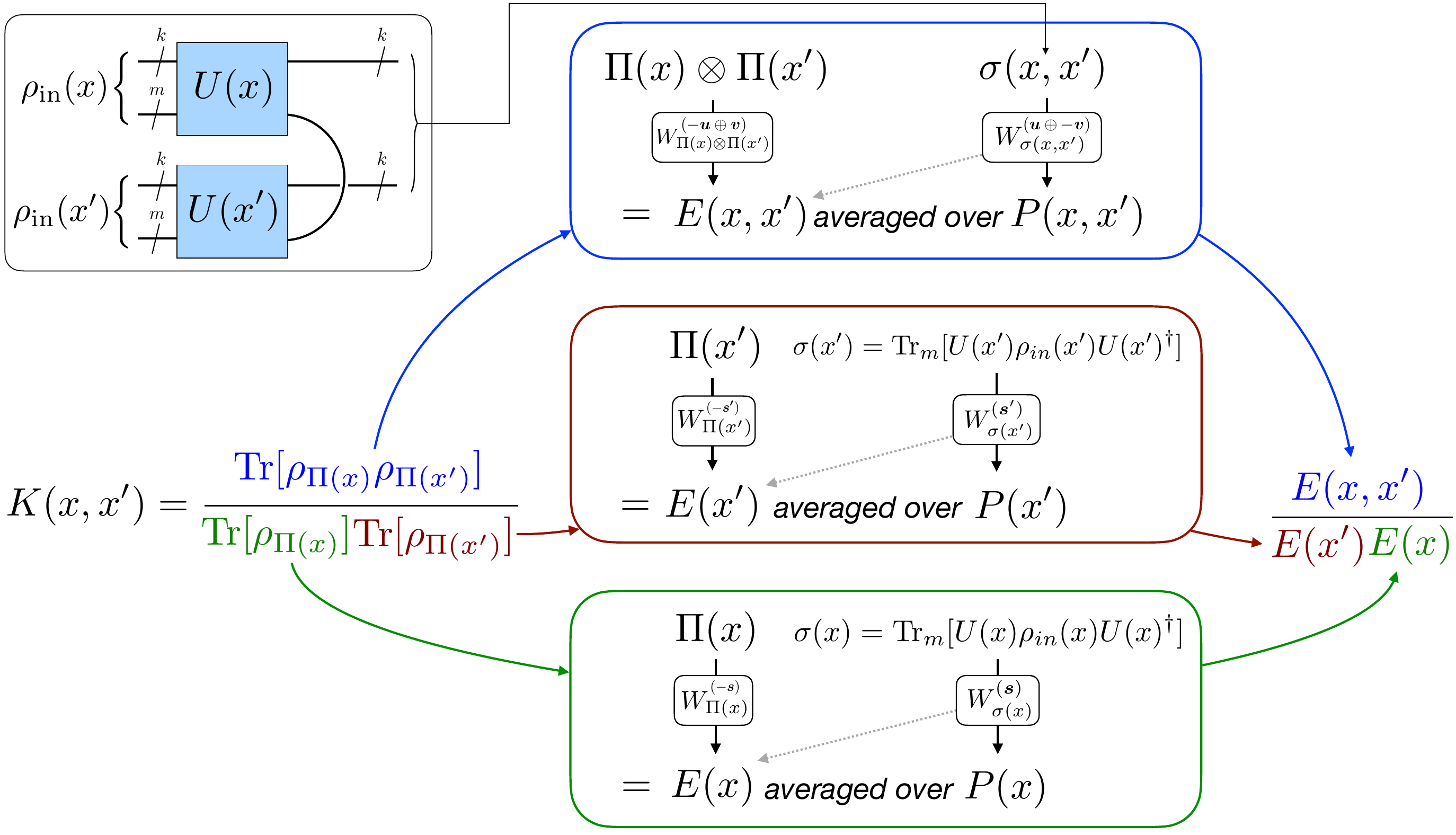}
		\caption{Classical estimation of quantum state overlap, following Algorithm~\ref{algo:MCQD-kern2}. The algorithm involves three separate estimations, each being similar to Algorithm~\ref{algo:MCQD-kern1} (see Fig.~\ref{fig:algo1}). The dashed arrows indicate that the negativity information of the corresponding \SPQD s $W_{\sigma(x)}^{(\bm s)}$, $W_{\sigma(x')}^{(\bm s')}$ and $W_{\sigma(x,x')}^{(\bm u\oplus-\bm v)}$ are also being used to define the corresponding estimator. The $k$-mode states $\sigma(x)$ and $\sigma(x')$ involved in the estimation of the denominator are defined as the partial trace of the last $m$ modes of the $(k+m)$-mode states $U(x)\rhoin(x)U(x)^\dag$ and $U(x')\rhoin(x')U(x')^\dag$, respectively, while the $(2k)$-mode state $\sigma(x,x')$ involved in the estimation of the numerator is defined as the partial overlap of the last $m$ modes of the $(k+m)$-mode states $U(x)\rhoin(x)U(x)^\dag$ and $U(x')\rhoin(x')U(x')^\dag$, as shown in the top-left circuit picture, where a partial transpose is omitted for brevity.}
		\label{fig:algo2}
	\end{center}
\end{figure*}

Note that in Algorithm~\ref{algo:MCQD-kern1}, $\rho(x)$ is used as a source of samples, while the estimator is mainly built from $\rho(x')$ (see Fig.~\ref{fig:algo1}).  On the other hand, in Algorithm~\ref{algo:MCQD-kern2}, both $U(x)\rhoin(x)U(x)^\dag$ and $U(x')\rhoin(x')U(x')^\dag$ are used as sources of samples, while the estimators are mainly built from $\Pi(x)$ and $\Pi(x')$ (see Fig.~\ref{fig:algo2}). The first method works well for generic cases, while the second method takes the advantage of the specific state preparation. We derive generic conditions for the efficiency of both methods in the next section, and we give concrete applications of these algorithms in Section~\ref{sec:examples}.

\subsection{Conditions for efficiency} 

The efficiency of the previous methods is based on assuming that all the \PQD s and their negative volume can be efficiently computed, and classical efficient sampling from the probability density is possible. However, if efficient computation of the \PQD s of the data encoding states is not possible, one may express the \PQD s in terms of the initial product states and the transition function describing the encoding circuit, as we show in Appendix~\ref{app:comp}.

We now present sufficient conditions for the efficiency of the above estimation procedures. 
Following either method, we obtain a classical estimate of the kernel function $K(x,x')$ with additive error $\epsilon=1/\mathrm{poly}(m)$, with an exponentially small probability of failure, using $N=\mathrm{poly}(m)$ samples, provided that both the negative volume $\mathcal N\big(W^{(\bm{t})}_\rho\big)$ and the range $\mathcal R(W^{(-\bm{t})}_A)$, bounding the range of the estimator $\mathcal R(E)$ as in Eq.~\eqref{eq:range-E}, are polynomial in $m$. Thus, to check this condition, one should look for ordering parameters such that $\mathcal R(E)$ is minimized.

Notice that having a nonnegative \PQD\ does not necessarily imply efficient classical estimation of the kernel function, since the other contributing range factor must also be considered in the minimization. In fact, the range of \PQD\ is also related to the negativity involved in kernel estimation experimental procedures. As $A$ is a bounded positive operator, it can be viewed as a POVM element of a two-outcome measurement $\{A,\bar{A}=I-A\}$. Therefore, estimation of $\Tr[\rho \bar A]$ results in estimation of the complement result of the measurement, $\Tr[\rho A]=1-\Tr[\rho \bar A]$.
However, since $W^{(-\bm t)}_{\bar A}(\bm{\mu})=1/\pi^m -W^{(-\bm t)}_A(\bm\mu)$, 
we can see that if the range of $W^{(-\bm t)}_A(\bm{\mu})$ is greater than $1/\pi^m$, then either it takes negative values, or the \PQD\ $W^{(-\bm t)}_{\bar A}(\bm{\mu})$ involved in the estimation of $\Tr[\rho \bar A]$ becomes negative. In other terms, the excess range of a \PQD\ is a signature of phase-space negativity contributing to the sample complexity of our classical algorithms.


To analyze how the sampling complexity depends on the ordering parameters, we derive an upper bound on the range 
$\mathcal R(W^{(-\bm{t})}_A)$. Let us first consider a single-mode operator $A_1$ with PQD $W^{(-t)}_{A_1}({\mu})$. The spectral decomposition of the single-mode Hermitian operator $\Delta^{(-t)}(\mu)$ is given by~\cite{CahGla69}
\begin{align}\label{eq:spec-single-Delta}
\Delta^{(-t)}(\mu)=\frac{2}{\pi(1+t)} \sum_{n=0}^{\infty}\! \bigg(\!\frac{t-1}{t+1}\!\bigg)^{n}\!\!
\ket{n,\mu}\bra{n,\mu},
\end{align}    
where $\ket{n,\mu}=D(\mu)\ket{n}$ are displaced number states. If $t\geq 0$, we have $(t-1)/(t+1)\leq1$. As the operator norm $\|A_1\|\leq1$, the maximum and the minimum values of the single-mode \negtPQD\ are determined by the largest and the smallest eigenvalues in Eq.~\eqref{eq:spec-single-Delta}:
\begin{equation}\label{eq:bound-single-Ws}
-\frac{2(1-t)}{\pi (1+t)^2}\leq  W^{(-t)}_{A_1}({\mu})\leq \frac{2}{\pi (1+t)}.
\end{equation}

Therefore, for $t=0$ the values of the Wigner function are between $-2/\pi$ and $2/\pi$, and as $t$ decreases the length of the interval becomes smaller, down to $1/\pi$ for the Husimi function when $t=1$. Notice that for $t<0$ the operator $\Delta^{(-t)}(\mu)$ has infinite eigenvalues, and therefore the values of \negtPQD are not bounded in general. 

Generalizing Eq.~\eqref{eq:bound-single-Ws} to the $m$-mode case, we have $\Delta^{(-\bm t)}(\bm\mu)=\otimes_{j=1}^{m}\Delta^{(-t_j)}(\mu_j)$,  and for $t_j\ge0$, 
\begin{equation}\label{eq:bound-multi-Ws}
\bigg(\!\frac{t_{\text{min}}-1}{t_{\text{min}}+1}\!\bigg)\!\prod_{j=1}^{m}\!\frac{2}{\pi(1+t_j)}\leq W^{(-\bm t)}_{A}(\bm{\gamma})\leq\! \prod_{j=1}^{m}\!\frac{2}{\pi(1+t_j)},\nonumber
\end{equation} 
where $t_{\text{min}}=\min_j t_j\ge0$ is the smallest ordering parameter. Using this inequality we find an upper bound for the interval length:
\begin{equation}\label{eq:gen-bound-range}
\mathcal{R}\big(W^{(-\bm t)}_{A}\big)\leq \frac{2}{t_{\text{min}}+1}\prod_{j=1}^{m}\!\frac{2}{t_j+1}\leq\left(\frac{2}{t_{\text{min}}+1}\right)^{m+1}\!\!\!\!\!\!\!\!\!\!.
\end{equation} 
This bound is $\smash{\mathcal{R}\big(W^{(-\bm t)}_{A}\big)\leq2^{m+1}}$ for the Wigner function, but for the Husimi function becomes $\mathcal{R}(W^{(-\bm I)}_{A})\leq1$. Therefore, for classical states with $\smash{\mathcal N\big(W^{(\bm{t})}_\rho\big)=1}$, an optimal choice for the ordering parameters is $\bm t=\bm I$ that gives $\mathcal{R}(E)\leq1$. This implies in particular that, using Algorithm~\ref{algo:MCQD-kern1}, the kernel function for classical states can be estimated efficiently classically.

For Algorithm~\ref{algo:MCQD-kern2}, the first three steps are instances of Algorithm~\ref{algo:MCQD-kern1} and their efficiency is subject to the previous considerations. Furthermore, we show in Appendix~\ref{app:ratio} that for any quantum-efficient encoding, combining three independent estimates of $\Tr[\rho_{\Pi(x)}]$, $\Tr[\rho_{\Pi(x')}]$ and $\Tr[\rho_{\Pi(x)}\rho_{\Pi(x')}]$ that are inverse-polynomially precise with exponentially small probability of failure provides an estimate of $K(x,x')=\Tr[\rho_{\Pi(x)}\rho_{\Pi(x')}]/(\Tr[\rho_{\Pi(x)}]\Tr[\rho_{\Pi(x')}])$ that is also inverse-polynomially precise with exponentially small probability of failure, which shows the efficiency of the last step of Algorithm~\ref{algo:MCQD-kern2}.

\section{Examples}
\label{sec:examples}

In this section, we apply our previous findings to various examples. Firstly, we illustrate the use of Algorithm ~\ref{algo:MCQD-kern1} in Section~\ref{sec:LON} for estimating quantum kernels based on output states of linear optical networks. Secondly, in Section~\ref{sec:G}, we show that there are cases where direct classical computation of the quantum kernel is efficient, such as for kernels based on Gaussian states. Finally, we illustrate the use of Algorithm~\ref{algo:MCQD-kern2} in Section~\ref{sec:part-meas-G} for a large class of quantum kernels based on non-Gaussian states obtained through the measurement of a subset of  modes of Gaussian states, including adaptive measurement strategies.


\subsection{Output states of linear optical networks} 
\label{sec:LON}


Linear optical networks (LON) are of particular interest because they can be simply realized, using passive optical elements such as beam splitters and phase shifters, yet they provide the underlying transformations in boson sampling problems~\cite{aaronson2011computational} that are believed to be classically hard to simulate. In boson sampling, one generates samples from the output probability distribution at the output of an LON, when single-photons are injected to the input. It is interesting to note, however, that output probabilities in boson sampling can be efficiently estimated classically~\cite{aaronson2011computational,Gurvits}. Also, classical algorithms have been recently proposed to approximate output probabilities of boson sampling and Gaussian boson sampling using the \SPQD s~\cite{Lim2023}.

Here, we consider a class of data-encoding states that are prepared by using LONs. In an ideal situation, a lossless LON is described by a unitary transformation $U_{\text{LON}}^\dagger(x)$ that can be defined by its action on the $m$-mode displacement operator, $U_{\text{LON}}^\dagger(x) D(\bm\xi) U_{\text{LON}}(x)=D(\bm\xi V(x))$, where $V(x)$ is an $m\times m$ unitary transfer matrix associated with the LON. Using this relation in Eqs.~(\ref{eq:Wigner}) and \eqref{eq:Delta-frame}, one can find the \SPQD\ of the output state of the LON $\rho(x)=U_{\text{LON}}(x)\rhoin\, U_{\text{LON}}^\dagger(x)$,
\begin{equation}\label{eq:PQD-lossless-LON}
W^{(\bm{s})}_{\rho(x)}\left(\bm{\alpha}\right)=W^{(V^\dagger(x)\bm{s}V(x))}_{\rhoin}\!\big(\bm\alpha V(x)\big),
\end{equation}
where $W^{(\bm{s})}_{\rhoin}(\bm\alpha)$ is the \SPQD\ of the input state $\rhoin$. This relation shows that, given the \SPQD\ of the input state, one can efficiently compute the \SPQD\ of the data-encoding states $\rho(x)$, and use Algorithm~\ref{algo:MCQD-kern1} to estimate the kernel function. 

In practice, however, any LON is lossy and hence cannot be described by a unitary transfer matrix. An interesting feature of our formalism is that it provides a practical way to take losses into account, as common sources of error, and check their effects on the negativity of the \PQD s, which is directly related to the efficiency of our kernel estimation algorithms. In general, any lossy LON can be modeled as a quantum operation consisting of single-mode loss channels sandwiched between two lossless LONs~\cite{rahimi2020situ}
\begin{equation}\label{eq:Lossy-LON}
\rho(x)=U_{\text{LON}}(x)\Lambda_{\bm\eta}\big(\tilde{U}_{\text{LON}}(x) \rhoin \tilde{U}^{\dagger}_{\text{LON}}(x)\big)U^{\dagger}_{\rm{LON}}(x).
\end{equation}
Here, $\Lambda_{\bm\eta}=\Lambda_{\eta_1}\otimes\dots\otimes\Lambda_{\eta_m}$ is an $m$-mode loss channel with parameters $0\leq\eta_j\leq1$. In Appendix~\ref{app:loss}, we have derived the relation between the \SPQD s of the output and input states of loss channels      
\begin{equation}\label{eq:PQD-under-loss}
W^{(\bm s)}_{\Lambda_{\bm\eta}(\rho)}(\bm{\alpha})=\frac{1}{\det \bm\eta} W^{(\bm t)}_{\rho} \big(\bm{\alpha}\bm\eta^{-1/2}\big),
\end{equation}
where $I-\bm s=\bm\eta(I-\bm{t})$ and $\bm\eta=\text{diag}(\eta_1,\dots,\eta_m)$ in the case of diagonal matrices of ordering parameters. This relation, together with Eq.~\eqref{eq:PQD-lossless-LON} enable us to investigate the effect of losses on the negative volume and range of the \PQD s of data-encoding states \eqref{eq:Lossy-LON}, and check when efficient estimation of the kernel function is possible. Note that an alternative approach presented in~\cite{rahimi2016sufficient} is to express the \PQD\ of the state~\eqref{eq:Lossy-LON} in terms of the \PQD\ of the input state and the Gaussian transition function associated with the lossy LON (see Eq.\eqref{eq:pqd-trans-func} of Appendix A). However, we emphasize that our classical estimation algorithms are inherently different from the classical sampling algorithms in~\cite{rahimi2016sufficient}, despite the similarities between the formalisms. For the sampling algorithms to work, all the \PQD s must be nonnegative, while our estimation algorithms always work and can be efficient even in the presence of negativity.

To further illustrate our formalism, let us assume that our initial state is product $\rhoin=\bigotimes_{k=1}^{m}\rho_k$ and losses can be modeled in terms of single-mode loss channels at the input of lossless LONs. We then consider examples of single-photon states and cat states as the initial states. Under these assumptions, Eq.~\eqref{eq:Lossy-LON} becomes $\rho(x)=U_{\text{LON}}(x)\bigotimes_{k=1}^{m} \tilde\rho_k U^{\dagger}_{\rm{LON}}(x)$ where $\tilde\rho_k=\Lambda_{\eta_k}(\rho_k)$. Note that losses in state preparation can also be incorporated into these loss channels at the input of LONs. In this scenario, assuming that $\bm s=s\bm{I}$ in Eq.~\eqref{eq:PQD-lossless-LON}, 
the \SPQD\ of the state $\rho(x)$ can be written as
\begin{align}\label{eq:sPQD-output-LON}
W^{(s\bm{I})}_{\rho(x)}\left(\bm{\alpha}\right)=\prod_{k=1}^{m}\! W^{(s)}_{\tilde\rho_k}\!\bigg(\sum_{j=1}^{m}\alpha_j V_{jk}(x)\bigg),
\end{align}  
where $W^{(s)}_{\tilde\rho_k}(\beta_k)$ is the \sPQD\ of the state injected into the $k$th input port of the lossless LON. In this case, the \SPQD s can always be efficiently computed and we can use our Algorithm~\ref{algo:MCQD-kern1}  
to estimate the kernel function
\begin{align*}
	K(x, x') =\pi^m\!\!\int\! \bdm{\bm{\alpha}} W^{(s\bm{I})}_{\rho(x)}\!\!\left(\bm{\alpha}\right) W^{(-s\bm{I})}_{\rho(x')}\!\!\left(\bm{\alpha}\right).
\end{align*}

By using Eq.~\eqref{eq:sPQD-output-LON} we can compute the probability distribution in Eq.~\eqref{eq:prob-dis}: 
\begin{align}\label{eq:prob-LON-example}
P(\bm\alpha)=\prod_{k=1}^{m}\!\frac{1}{\mathcal{N}\big(W^{(s)}_{\tilde\rho_k}\big)}\bigg\vert W^{(s)}_{\tilde\rho_k}\!\bigg(\sum_{j=1}^{m}\alpha_j V_{jk}(x)\!\bigg)\bigg\vert.
\end{align}
Here, $\mathcal{N}\big(W^{(s)}_{\tilde\rho_k}\big)=\int \text{d}^2\bm\beta_k \vert W^{(s)}_{\tilde\rho_k}(\beta_k) \vert$. Notice that $\mathcal{N}\big(W^{(s\bm{I})}_{\rho(x)}\big)=\prod_{k=1}^{m} \mathcal{N}\big(W^{(s)}_{\tilde\rho_k}\big)$ since lossless LONs preserve the negative volume and hence the non-classicality of quantum states~\cite{albarelli2018resource}. We can also efficiently generate samples from the probability distribution~\eqref{eq:prob-LON-example} by first sampling the components of $N$ complex vectors $\bm\beta_1, \dots, \bm\beta_{N}$ from individual probability distributions $\vert W^{(s)}_{\tilde\rho_k}(\beta_k) \vert/\mathcal{N}\big(W^{(s)}_{\tilde\rho_k}\big)$, $k=1, \dots, N$, and then calculating $\bm\alpha_j=\bm\beta_j V^\dagger(x) $. Given these samples, we then obtain the corresponding values for the \SPQD\ of $\rho(x')$,  $W^{(-s\bm{I})}_{\rhoin}\!\big(\bm{\alpha}_j V(x')\big)$, and through Eq.~(\ref{eq:estimator}) calculate the sample mean $\frac{1}{N}\sum_{j=1}^{N} E(\bm\alpha_j)$.

The sampling complexity for this class of states is described by
\begin{align}\label{eq:M-bound-LON}
	\begin{split}
\mathcal{R}(E)&\leq 2 \pi^m \mathcal{N}\big(W^{(s\bm{I})}_{\rho(x)}\big) \max_{\bm\alpha} \big\vert W^{(-s\bm I)}_{\rho(x')}(\bm\alpha)\big\vert\\
&=2\prod_{k=1}^{m} \left(\pi  \mathcal{N}\big(W^{(s)}_{\tilde\rho_k}\big) \max_{\alpha} \big\vert W^{(-s)}_{\tilde\rho_k}(\alpha)\big\vert\right)\!,
	\end{split}
\end{align} 
which is completely independent of the data-encoding LON unitary operations and depends only on the \sPQD\ of the input states. Therefore, for a given input state $\rhoin=\bigotimes_{k=1}^{m}\rho_k$ and loss parameters, one can minimize $\mathcal R(E)$ by finding an optimal ordering parameter $s$ and determine the scaling for the number of required samples. Notice, however, that one could use a more optimal estimation procedure with different ordering parameters for each input mode, in which case this would depend on the LON description. Using Eq.~\eqref{eq:M-bound-LON}, it is easy to verify that if all states $\tilde\rho_k$ are classical, then for the optimal choice of $s=1$ we have $\mathcal{R}(E)=1$.

\subsubsection{Example 1: Single-photon states}
Let us consider the case where all the input states are single-photon states $\rho_k=\ketbra{1}$ in the above formalism. Hence, the input states of lossless LONs are $\tilde\rho_k=\Lambda_{\eta_k}(\ketbra{1})=(1-\eta_k)\ketbra{0}+\eta_k\ketbra{1}$. Using the single-mode version of Eq.~\eqref{eq:Delta-frame}, the \sPQD\ of $\tilde\rho_k$ is given by
\begin{align*}
	W^{(s)}_{\tilde\rho_k}(\alpha)\!=\!\frac{2(1-s)(1-s-2\eta_k)+8\eta_k|\alpha|^2}{\pi (1-s)^3}e^{-2|\alpha|^2/(1-s)},
\end{align*}
which is non-negative if $s\leq 1-2\eta_k$. Moreover, if $s> 1-2\eta_k$, the negative volume is given by
\begin{equation}\label{eq:neg-rho_k}
 \mathcal{N}\big(W^{(s)}_{\tilde\rho_k}\big)=\frac{4 \eta_k}{1-s}\exp\!\bigg(\!\frac{1-s-2\eta_k}{2 \eta_k}\!\bigg)-1.
\end{equation}
To examine the bound (\ref{eq:M-bound-LON}), we should also consider the maximum value of $	W^{(-s)}_{\tilde\rho_k}(\alpha)$. This function has two extremal values at $|\alpha_0|^2=0$ and $|\alpha_1|^2=(1+s)(4\eta_k-1-s)/(2\eta_k)$, which are given by
\begin{align*}
	W^{(s)}_{\tilde\rho_k}(0)=\frac{2 (1+s-2 \eta_k)}{\pi  (1+s)^2},
\end{align*}
and 
\begin{align*}
	W^{(s)}_{\tilde\rho_k}(\alpha_1)=\frac{4 \eta}{\pi(1+s)^2} \exp\!\bigg(\!\frac{s+1-4\eta_k}{2 \eta_k}\!\bigg).
\end{align*}
Therefore, using these expressions and Eq.~(\ref{eq:neg-rho_k}) in Eq.~(\ref{eq:M-bound-LON}), and then optimizing over $s$, we can find how the sample complexity scales with the number of modes.  

As an example, assuming that $\eta_k=1/2$ for all $k$, we can verify that $s=0$ is an optimal ordering parameter since $\mathcal{N}^{(0)}_{\rho_k}=1$ and $0\leq W^{(0)}_{\rho_k}(\alpha)\leq 2/(e \pi)$. Therefore, in this case, $\mathcal{R}(E)\leq 1$ and the kernel function can be estimated efficiently to within an additive error $\epsilon=1/\text{poly}(m)$ with an exponentially small probability of failure.  Numerical analysis can be utilized to handle other values of $\eta_k$. For instance, by Eq.~(\ref{eq:M-bound-LON}), it can be seen that for $\eta=0.85$ and the ordering parameter $s=0.3$,  $\mathcal{R}(E)\leq 2$, independent of the number of modes, and hence our estimation algorithm is still efficient. 

Note that due to their generality, phase-space methods may be outperformed by other classical simulation techniques. For instance, we derive a variant of Gurvits' algorithm for estimating the permanent~\cite{Gurvits} in Appendix~\ref{app:lossyGurvits} which allows us to perform efficient classical estimation of quantum kernels based on lossy photonic states for \textit{any} loss parameter.

\subsubsection{Example 2: Cat states}
Another class of states that we consider as the input states for LONs are the cat states
\begin{equation}
    \ket{\text{cat}_k}=\frac{\ket{\gamma_k}+\ket{-\gamma_k}}{\sqrt{2+2e^{-2|\gamma_k|^2}}},
\end{equation}
where $\ket{\gamma_k}$ denotes a coherent state. By using Eq.~\eqref{eq:Wigner}, the \sPQD\ of this state is given by
\begin{align}
           W^{(s)}_{\ket{\text{cat}_k}}(\alpha)&=\frac{1}{\pi(1-s)(1+e^{-2|\gamma_k|^2})} \Bigg( e^{-2\frac{|\alpha+\gamma_k|^2}{1-s}}\nonumber\\
           &+e^{-2\frac{|\alpha-\gamma_k|^2}{1-s}}\\
           &+2e^{-2|\gamma|^2} \Re\bigg(e^{-2\frac{(\alpha+\gamma_k)(\alpha^*-\gamma_k^*)}{1-s}}\bigg) \!\Bigg), \nonumber
\end{align}
 where $\Re$ denotes the real part of the expression. Using this equation and Eq.~\eqref{eq:PQD-under-loss}, we can then compute the \sPQD\ of the states after loss channels, $\tilde\rho_k=\Lambda_{\eta_k}(\ketbra{\text{cat}_k})$,
\begin{align}
            &W^{(s)}_{\tilde\rho_k}(\alpha)=\frac{1}{\pi(1-s)(1+e^{-2|\gamma_k|^2})} \Bigg(e^{-2\frac{|\alpha+\sqrt{\eta_k}\gamma_k|^2}{1-s}}\nonumber\\
           &\quad +e^{-2\frac{|\alpha-\sqrt{\eta_k}\gamma_k|^2}{1-s}}\\
           &\quad +2e^{-2|\gamma_k|^2} \Re\bigg(e^{-2\frac{\left(\alpha+\sqrt{\eta_k}\gamma_k\right)\!\left(\alpha^*-\sqrt{\eta_k}\gamma_k^*\right)}{1-s}}\bigg)\!\!\Bigg).\nonumber 
\end{align}
Given the parameters for input cat states and losses, one can compute the upper bound on the range of the estimator $\mathcal{R}(E)$ in Eq.~(\ref{eq:M-bound-LON}), and find an optimal ordering parameter $s$ by using this expression. 

For example, numerical analysis shows that if $\gamma_k=4$ and $\eta_k=0.8$ for all $k$, then $\pi  \mathcal{N}(W^{(s)}_{\tilde\rho_k}) \max_{\alpha} \vert W^{(-s)}_{\tilde\rho_k}(\alpha)\vert <1$ for $s=0.1$. Therefore, in this case, we have $\mathcal{R}(E)<2$, independent of the number of modes, and our Algorithm~\ref{algo:MCQD-kern1} can be used to estimate the kernel function efficiently.  


\subsection{Gaussian states} 
\label{sec:G}

We now focus on Gaussian states. Choosing the parameter $\bm s=\bm 0$ provides an efficient classical estimation of the quantum kernel through Algorithm~\ref{algo:MCQD-kern1}. However, instead of using the estimation algorithm we described, we can check whether the kernel function~\eqref{eq:kern-sPQD} can be computed directly using conventional methods for computing integrals. Indeed, for Gaussian states the corresponding quantum kernel functions can be computed exactly analytically. 

Gaussian states have Gaussian Wigner functions that can be described in terms of the mean values $\bar{\bm r}=\Tr[\rho\, \bm{r}]$, where $\bm{r}=(q_1, p_1,\dots, q_m, p_m)^T$ is the vector of canonical operators $q_j=(a_j+a_j^\dagger)/\sqrt{2}$ and $p_j=i(a_j^\dagger-a_j)/\sqrt{2}$, and the covariance matrix $\Sigma_{jk}=\Tr[\rho(\bm{r}_j\bm{r}_k+\bm{r}_k\bm{r}_j)]/2-\bar{\bm r}_j\bar{\bm r}_k$. Indeed, for such a Gaussian state $\rho$ we have~\cite{ferraro2005}:
\begin{equation}\label{eq:s-PQD-G}
W^{(\bm s)}_\rho\left(\bm{\alpha}\right)=\frac{ e^{-\frac 12 (\bm{\alpha}- \bar{\bm r}) (\Sigma-\bm s\oplus\bm s)^{-1} (\bm{\alpha}- \bar{\bm r})^\top}}{(2\pi)^m\sqrt{\det(\Sigma-\bm s\oplus\bm s)}},
\end{equation}
for all $\bm s$ such that $\Sigma-\bm s\oplus\bm s$ is positive definite, where we used that \SPQD s are related to the Wigner function by a Gaussian convolution. For single-mode Gaussian states, the minimal valid choice for $\tau:=\frac12(1-s)\in[0,\frac12]$ is known as the non-classical depth~\cite{lee1991measure}. This definition readily extends to the multimode setting (we use a simplified version of the general definition in~\cite{Sabapathy2016multimode}):

\begin{defi}[Non-classical depth]\label{def:nc-depth} The non-classical depth of a quantum state $\rho$ is the minimal value $\tau=\frac12(1-s)\in[0,1]$ such that the \SPQD\ of the state $\rho$ is non-negative for $\bm s=s\bm I$.
\end{defi}


\noindent By Eq.~(\ref{eq:s-PQD-G}), for multimode Gaussian states the minimal eigenvalue of the covariance matrix encodes this information. Using Eq.~\eqref{eq:kern-sPQD} for $\bm s=\bm 0$ one can verify that the kernel function is given by
\begin{align}
K(x,x')\!=\!\frac{e^{-\frac{1}{2}(\bar{\bm{r}}-\bar{\bm{r}}')(\Sigma(x)+\Sigma(x'))^{-1}(\bar{\bm{r}}-\bar{\bm{r}}')^{T}}}{2^{-m}\sqrt{\det\left(\Sigma(x)+\Sigma(x')\right)}}.\nonumber
\end{align}
Thus, for Gaussian states we do not really need our Monte Carlo-based method to estimate the kernel function, and quantum machine learning protocols that use data-encoding Gaussian states can be efficiently simulated by classical algorithms~\cite{SchuldPRA2020}. We note that it is strongly believed that sampling from the photon-counting probability distributions for Gaussian states cannot be simulated efficiently classically~\cite{Lund2015,RahimiComplexity2015,Hamilton2017}. Thus, although computing the kernel function for such states is easy, sampling from their photon-counting distribution is hard.


\subsection{Partially measured Gaussian states} 
\label{sec:part-meas-G}

Given the limitations of Gaussian states for quantum kernel methods, we can ask whether non-Gaussian states can be helpful. A standard way to engineer a non-Gaussian state is to perform non-Gaussian measurements on a subset of the modes of a Gaussian state (see~\cite{lvovsky2020production} and references therein). We thus consider the special case of Eq.~(\ref{eq:norm-state}) consisting of quantum states $\rho(x)$ prepared by measuring some of the output modes of a multimode Gaussian state, i.e.,
\begin{equation}\label{eq:norm-state-G}
x\mapsto\rho(x)=\frac{\mathrm{Tr}_k[(\Pi(x)\otimes\mathbb I_m)\rho_G(x)]}{\mathrm{Tr}[(\Pi(x)\otimes\mathbb I_m)\rho_G(x)]},
\end{equation}
where $\rho_G(x)=U(x)\rhoin(x)U(x)^\dag$ is an $(k+m)$-mode Gaussian state and $\Pi(x)=\bigotimes_{j=1}^k\Pi_j(x)$ is a tensor product of (possibly non-Gaussian) POVM elements. Recall that a quantum-efficient encoding refers to the fact that such states may be efficiently prepared using a quantum computer, a property which can be summarized here by $\mathrm{Tr}[(\Pi(x)\otimes\mathbb I_m)\rho_G(x)]\ge1/(\mathrm{poly}(m) )$.

Using our kernel estimation formalism and Algorithm~\ref{algo:MCQD-kern2} in particular, we show that, when kernel estimation is quantum-efficient,  classical kernel estimation is also efficient whenever either the number of measured modes or the non-classicality of the Gaussian states involved is too small:

\begin{theo}\label{th:part-meas-G}
    For any classical data $x$, let $\rho(x)$ be a quantum state encoding over $m$ modes obtained by performing a measurement of the first $k$ modes of a $(k+m)$-mode Gaussian state $\rho_G(x)$, as in Eq.~(\ref{eq:norm-state-G}). Let $\tau(x)$ denote the nonclassical depth of $\rho_G(x)$ (see Definition~\ref{def:nc-depth}) and let $\tau(x,x')=\max(\tau(x),\tau (x'))\in[0,\frac12]$. Then, assuming that the encoding is quantum-efficient, Algorithm~\ref{algo:MCQD-kern2} provides an estimate of the quantum kernel $K(x,x')=\Tr[\rho(x)\rho(x')]$ with additive precision $\epsilon$ and success probability $1-\delta$ in time
    \begin{equation}
         O\left(\frac{\log(\frac2\delta)}{\epsilon^2(1-\tau(x,x'))^{4k+2}} \mathrm{poly}(m)\right),
    \end{equation}
    %
    In particular, this provides an efficient classical algorithm for quantum kernel estimation whenever $k=O(\log m)$ or else $\tau(x,x')=O(\log m/k)$.
\end{theo}

\noindent We give a proof of this theorem in Appendix~\ref{app:part-meas-G}, which combines a careful analysis of the time complexity of Algorithm~\ref{algo:MCQD-kern2} together with new properties of the non-classical depth of Gaussian states.

A nontrivial consequence of Theorem~\ref{th:part-meas-G} is that the efficiency of the classical simulation is independent of the non-Gaussianity of the measurements: even though these can inject a lot of negativity in the prepared state, as measured by the negative volume, this negativity does not affect the classical simulability through Algorithm~\ref{algo:MCQD-kern2}, because the POVM elements only contribute to defining the estimators in Algorithm~\ref{algo:MCQD-kern2} and their non-Gaussianity does not substantially change the range of these estimators.
In particular, even when making very non-Gaussian measurements (such as detecting many photons), classical simulation is always efficient as long as only a few modes $k=O(\log m)$ are measured. 

What about when the number of measured modes is larger? There, Theorem~\ref{th:part-meas-G} shows that classical simulation is still efficient when the Gaussian state being measured has small non-classicality, as quantified by the non-classical depth (see Definition~\ref{def:nc-depth}). This notion of non-classicality is directly related to the amount of thermal noise necessary to make the state fully classical~\cite{lee1991measure}, i.e., with non-negative $P$ function, and for a Gaussian state it is related to the minimum eigenvalue of its covariance matrix, see Eq.~(\ref{eq:s-PQD-G}). This quantity can in turn be bounded by a function of the squeezing parameters and the symplectic eigenvalues encoding the impurity of the corresponding Gaussian state. 

For illustration purposes, let us consider $\rho_G:=U\rhoin U^\dag$ with $\rhoin$ being a tensor product of identical single-mode thermal states $\nu$, and $U$ being a Gaussian unitary operator. By virtue of the Euler (or Bloch--Messiah) decomposition~\cite{ferraro2005} we may write $U=DVSV'$ with $D$ being a tensor product of single-mode displacement operators, $S$ being a tensor product of single-mode squeezing operators, and $V, V'$ being passive linear transforms describing the action of lossless LONs. Hence, $\rho_G=DVSV'\nu^{\otimes(k+m)}{V'}^\dag S^\dag V^\dag D^\dag=DVS\nu^{\otimes(k+m)}S^\dag V^\dag D^\dag$, where we used the fact that a tensor product of identical single-mode thermal states is invariant under lossless LONs.
This state has the following covariance matrix~\cite{ferraro2005}:
\begin{equation}
	O_V\begin{pmatrix}\frac1\mu\Delta^2&0\\0&\frac1\mu\Delta^{-2}\end{pmatrix}O_V^\top,
\end{equation}
where $O_V$ is the symplectic orthogonal matrix associated to the LON $V$, $\Delta$ is a diagonal matrix containing the squeezing parameters and $\mu$ is the purity of the single-mode thermal state $\nu$. Writing $\frac1r$ for the minimal squeezing parameter smaller than $1$, with $r\ge1$, the non-classical depth of the state is given by $\tau=\frac12(1-\frac1{r^2\mu})$ by Eq.~(\ref{eq:s-PQD-G}), and the condition $\tau=O(\log m/k)$ in Theorem~\ref{th:part-meas-G} implies that classical estimation of quantum kernels based on partially measured Gaussian states of the form of $\rho_G$ is efficient whenever the squeezing $r$ or the purity $\mu$ are too small.

Theorem~\ref{th:part-meas-G} also allows us to investigate the effect of lossy state preparation: with Eq.~(\ref{eq:PQD-under-loss}) and Definition~\ref{def:nc-depth}, uniform losses of transmissivity $\eta$ over all modes map the non-classical depth from $\tau$ to $\eta\,\tau$, in which case the classical estimation provided by Algorithm~\ref{algo:MCQD-kern2} is efficient whenever $k=O(\log m)$ or else $\eta\,\tau(x,x')=O(\log m/k)$.

Finally, up to taking mixtures, our results on partially measured Gaussian states also cover the case of quantum states prepared by Gaussian computations, together with \textit{adaptive} measurements, i.e., intermediate measurements whose outcome can drive the rest of the computation. In particular, we show that classical estimation of the corresponding quantum kernel functions is efficient under the conditions of Theorem~\ref{th:part-meas-G}, if the number of possible adaptive measurement outcomes is small enough (see Appendix~\ref{app:AGBS} for details).


\section{Conclusion and outlook}
\label{sec:concl}

We have introduced a framework based on phase-space quasi-probability distributions for the classical estimation of quantum kernel functions in machine learning. Our sufficient conditions for efficient classical simulation are based on negative volume, non-classical depth, and excess range of quasi-probability distributions, and identify phase-space negativity as an essential resource for achieving computational advantages in quantum machine learning with kernel methods. Our formalism can also be used to investigate the effect of errors and imperfections in quantum machine learning devices by examining their impact on the negativity of quasi-probability distributions.
	
By considering various examples based on variants of the boson sampling model, we have showcased how sampling from the output probability distribution of a quantum circuit can be classically hard, yet supervised machine learning using the same circuit can be classically efficient. 

Moreover, we have identified a subtle interplay between the quantum computational resources at hand: if no phase-space negativity is involved, and in the case of Gaussian measurements in particular, quantum kernels based on partially measured Gaussian states can be efficiently estimated classically using our Algorithm~\ref{algo:MCQD-kern1} (see Section~\ref{sec:G}). When phase-space negativity is present in the measurements, then our Algorithm~\ref{algo:MCQD-kern2} still allows for efficient classical estimation of the corresponding quantum kernels as long as the number of measured modes is small enough or the non-classical depth of the Gaussian states involved is small enough (see Section~\ref{sec:part-meas-G}). In other terms, quantum computational advantage for quantum kernel estimation is only possible in this setting by combining Gaussian non-classical resources (squeezing) and non-Gaussian resources (phase-space negative volume). A similar situation arises in the context of Gaussian boson sampling~\cite{hamilton2017gaussian}, where squeezing is a necessary ingredient together with non-Gaussianity for quantum computational advantage through sampling~\cite{chabaud2021holomorphic,chabaud2023resources}.
	
Our results could be extended in a few directions. 
The sample complexities of our classical simulation algorithms are naturally expressed using non-classical measures relating to phase-space negativity; it would be interesting to relate these measures to other existing ones such as quadrature coherence scale \cite{hertz2020quadrature}, which provides an estimation of
the distance to the set of classical states, or stellar rank \cite{chabaud2020stellar}, according to which a classification of bosonic kernels for quantum machine learning was recently derived \cite{henderson2024general}. 
Another direction could be to use our framework to analyze in more details how specific imperfections in implementations of quantum kernel methods, and the SWAP test in particular \cite{gao2018programmable,gan2020hybrid,vcernotik2024swap}, could ease classical simulability.
One could also apply the presented framework to the case of non-linear optical approaches such as the optical Ising machine~\cite{Ising} or Kerr-based kernels \cite{dehdashti2024enhancing}. 
Moreover, it would be interesting to generalize the presented approach to the case of discrete-variable quasi-probability distributions using frame theory~\cite{frame,Ferrie09}. 
	
\medskip

\noindent
{\it Acknowledgements.} U.C.\ acknowledges inspiring discussions with M.\ Walschaers, M.\ Frigerio, F.\ Arzani, J.\ Davis, M.\ Garnier, H.\ Thomas, P.E.\ Emeriau, S.\ Mehraban and D.\ Hangleiter. R.G.\ acknowledges discussion with C.\ Simon at the early stage of this project. 
U.C.\ acknowledges funding from the European Union's Horizon Europe Framework Programme (EIC Pathfinder Challenge project Veriqub) under Grant Agreement No.\ 101114899.

\bibliographystyle{linksen}
\bibliography{biblio}

\providecommand{\href}[2]{#2}\begingroup\raggedright\begin{thebibliography}{10}

\bibitem{preskill}
J.~Preskill, ``Quantum {C}omputing in the {NISQ} era and beyond,''
  \href{http://dx.doi.org/10.22331/q-2018-08-06-79}{{\em {Quantum}} {\bfseries
  2}, 79 (2018)}.

\bibitem{arute2019quantum}
F.~Arute, K.~Arya, R.~Babbush, D.~Bacon, J.~C. Bardin, R.~Barends, R.~Biswas,
  S.~Boixo, F.~G. Brandao, D.~A. Buell, {\em et al.}, ``Quantum supremacy using
  a programmable superconducting processor,''
  \href{http://dx.doi.org/10.1117/12.2603523}{{\em Nature} {\bfseries 574},
  505--510 (2019)}.

\bibitem{RMP14}
I.~M. Georgescu, S.~Ashhab, and F.~Nori, ``Quantum simulation,''
  \href{http://dx.doi.org/10.1103/RevModPhys.86.153}{{\em Rev. Mod. Phys.}
  {\bfseries 86}, 153--185 (2014)}.

\bibitem{qchem}
A.~J. McCaskey, Z.~P. Parks, J.~Jakowski, S.~V. Moore, T.~D. Morris, T.~S.
  Humble, and R.~C. Pooser, ``Quantum chemistry as a benchmark for near-term
  quantum computers,'' \href{http://dx.doi.org/10.1038/s41534-019-0209-0}{{\em
  npj Quantum Information} {\bfseries 5}, 99 (2019)}.

\bibitem{biamonte2017quantum}
J.~Biamonte, P.~Wittek, N.~Pancotti, P.~Rebentrost, N.~Wiebe, and S.~Lloyd,
  ``Quantum machine learning,''
  \href{http://dx.doi.org/10.1038/nature23474}{{\em Nature} {\bfseries 549},
  195--202 (2017)}.

\bibitem{Farhi}
E.~Farhi and H.~Neven, ``Classification with quantum neural networks on near
  term processors,'' \href{http://arxiv.org/abs/arXiv:1802.06002}{{\ttfamily
  arXiv:1802.06002}}.

\bibitem{terhal2002adaptive}
B.~M. Terhal and D.~P. DiVincenzo, ``Adaptive quantum computation, constant
  depth quantum circuits and arthur-merlin games,''
  \href{http://dx.doi.org/10.26421/QIC4.2-5}{{\em Quantum Inf. Comput.}
  {\bfseries 4}, 134--145 (2004)}.

\bibitem{bremner2016average}
M.~J. Bremner, A.~Montanaro, and D.~J. Shepherd, ``Average-Case Complexity
  Versus Approximate Simulation of Commuting Quantum Computations,''
  \href{http://dx.doi.org/10.1103/PhysRevLett.117.080501}{{\em Phys. Rev.
  Lett.} {\bfseries 117}, 080501 (2016)}.

\bibitem{aaronson2011computational}
S.~Aaronson and A.~Arkhipov, ``The Computational Complexity of Linear Optics,''
  \href{http://dx.doi.org/10.4086/toc.2013.v009a004}{{\em Theory of Computing}
  {\bfseries 9}, 143--252 (2013)}.

\bibitem{Science}
S.~Bravyi, D.~Gosset, and R.~K\"onig, ``Quantum advantage with shallow
  circuits,'' \href{http://dx.doi.org/10.1126/science.aar3106}{{\em Science}
  {\bfseries 362}, 308--311 (2018)}.

\bibitem{CahGla69}
K.~E. Cahill and R.~J. Glauber, ``Density Operators and Quasiprobability
  Distributions,'' \href{http://dx.doi.org/10.1103/PhysRev.177.1882}{{\em Phys.
  Rev.} {\bfseries 177}, 1882--1902 (1969)}.

\bibitem{Spekkens}
R.~W. Spekkens, ``Negativity and Contextuality are Equivalent Notions of
  Nonclassicality,''
  \href{http://dx.doi.org/10.1103/PhysRevLett.101.020401}{{\em Phys. Rev.
  Lett.} {\bfseries 101}, 020401 (2008)}.

\bibitem{mari2012positive}
A.~Mari and J.~Eisert, ``Positive Wigner Functions Render Classical Simulation
  of Quantum Computation Efficient,''
  \href{http://dx.doi.org/10.1103/PhysRevLett.109.230503}{{\em Phys. Rev.
  Lett.} {\bfseries 109}, 230503 (2012)}.

\bibitem{veitch2013efficient}
V.~Veitch, N.~Wiebe, C.~Ferrie, and J.~Emerson, ``Efficient simulation scheme
  for a class of quantum optics experiments with non-negative Wigner
  representation,'' \href{http://dx.doi.org/10.1088/1367-2630/15/1/013037}{{\em
  New Journal of Physics} {\bfseries 15}, 013037 (2013)}.

\bibitem{rahimi2016sufficient}
S.~Rahimi-Keshari, T.~C. Ralph, and C.~M. Caves, ``Sufficient Conditions for
  Efficient Classical Simulation of Quantum Optics,''
  \href{http://dx.doi.org/10.1103/PhysRevX.6.021039}{{\em Phys. Rev. X}
  {\bfseries 6}, 021039 (2016)}.

\bibitem{Stahlke}
D.~Stahlke, ``Quantum interference as a resource for quantum speedup,''
  \href{http://dx.doi.org/10.1103/PhysRevA.90.022302}{{\em Phys. Rev. A}
  {\bfseries 90}, 022302 (2014)}.

\bibitem{pashayan2015estimating}
H.~Pashayan, J.~J. Wallman, and S.~D. Bartlett, ``Estimating Outcome
  Probabilities of Quantum Circuits Using Quasiprobabilities,''
  \href{http://dx.doi.org/10.1103/PhysRevLett.115.070501}{{\em Phys. Rev.
  Lett.} {\bfseries 115}, 070501 (2015)}.

\bibitem{mengoni2019kernel}
R.~Mengoni and A.~Di~Pierro, ``Kernel methods in quantum machine learning,''
  \href{http://dx.doi.org/10.1007/s42484-019-00007-4}{{\em Quantum Machine
  Intelligence} {\bfseries 1}, 65--71 (2019)}.

\bibitem{ghobadi2021nonclassical}
R.~Ghobadi, ``Nonclassical kernels in continuous-variable systems,''
  \href{http://dx.doi.org/10.1103/PhysRevA.104.052403}{{\em Physical Review A}
  {\bfseries 104}, 052403 (2021)}.

\bibitem{bohmann2020phase}
M.~Bohmann and E.~Agudelo, ``Phase-space inequalities beyond negativities,''
  \href{http://dx.doi.org/10.1103/PhysRevLett.124.133601}{{\em Physical Review
  Letters} {\bfseries 124}, 133601 (2020)}.

\bibitem{kenfack2004negativity}
A.~Kenfack and K.~{\.Z}yczkowski, ``Negativity of the {W}igner function as an
  indicator of non-classicality,''
  \href{http://dx.doi.org/10.1088/1464-4266/6/10/003}{{\em Journal of Optics B:
  Quantum and Semiclassical Optics} {\bfseries 6}, 396 (2004)}.

\bibitem{albarelli2018resource}
F.~Albarelli, M.~G. Genoni, M.~G.~A. Paris, and A.~Ferraro, ``Resource theory
  of quantum non-Gaussianity and Wigner negativity,''
  \href{http://dx.doi.org/10.1103/PhysRevA.98.052350}{{\em Phys. Rev. A}
  {\bfseries 98}, 052350 (2018)}.

\bibitem{lee1991measure}
C.~T. Lee, ``Measure of the nonclassicality of nonclassical states,''
  \href{http://dx.doi.org/10.1103/PhysRevA.44.R2775}{{\em Physical Review A}
  {\bfseries 44}, R2775 (1991)}.

\bibitem{Sabapathy2016multimode}
K.~K. Sabapathy, ``Process output nonclassicality and nonclassicality depth of
  quantum-optical channels,''
  \href{http://dx.doi.org/10.1103/PhysRevA.93.042103}{{\em Phys. Rev. A}
  {\bfseries 93}, 042103 (2016)}.

\bibitem{schuld2021supervised}
M.~Schuld and F.~Petruccione, {\em Quantum Models as Kernel Methods},
  \href{http://dx.doi.org/10.1007/978-3-030-83098-4_6}{pp.~217--245}.
\newblock Springer International Publishing, Cham, 2021.

\bibitem{scholkopf2001generalized}
B.~Sch{\"o}lkopf, R.~Herbrich, and A.~J. Smola, ``A Generalized Representer
  Theorem,'' \href{http://dx.doi.org/10.1007/3-540-44581-1_27}{in {\em
  Computational Learning Theory}, D.~Helmbold and B.~Williamson, eds.},
  pp.~416--426.
\newblock Springer Berlin Heidelberg, Berlin, Heidelberg, 2001.

\bibitem{hofmann2008kernel}
T.~Hofmann, B.~Sch{\"o}lkopf, and A.~J. Smola, ``Kernel methods in machine
  learning,'' \href{http://dx.doi.org/10.1214/009053607000000677}{{\em The
  Annals of Statistics} {\bfseries 36}, 1171--1220 (2008)}.

\bibitem{havlivcek2019supervised}
V.~Havl{\'\i}{\v{c}}ek, A.~D. C{\'o}rcoles, K.~Temme, A.~W. Harrow, A.~Kandala,
  J.~M. Chow, and J.~M. Gambetta, ``Supervised learning with quantum-enhanced
  feature spaces,'' \href{http://dx.doi.org/10.1038/s41586-019-0980-2}{{\em
  Nature} {\bfseries 567}, 209--212 (2019)}.

\bibitem{schuld2019quantum}
M.~Schuld and N.~Killoran, ``Quantum Machine Learning in Feature {H}ilbert
  Spaces,'' \href{http://dx.doi.org/10.1103/PhysRevLett.122.040504}{{\em Phys.
  Rev. Lett.} {\bfseries 122}, 040504 (2019)}.

\bibitem{buhrman2001quantum}
H.~Buhrman, R.~Cleve, J.~Watrous, and R.~De~Wolf, ``Quantum fingerprinting,''
  \href{http://dx.doi.org/10.1103/PhysRevLett.87.167902}{{\em Physical review
  letters} {\bfseries 87}, 167902 (2001)}.

\bibitem{HilOCo84}
M.~Hillery, R.~O'Connell, M.~Scully, and E.~Wigner, ``Distribution functions in
  physics: Fundamentals,''
  \href{http://dx.doi.org/https://doi.org/10.1016/0370-1573(84)90160-1}{{\em
  Physics Reports} {\bfseries 106}, 121--167 (1984)}.

\bibitem{Chakh2017}
R.~P. Rundle, P.~W. Mills, T.~Tilma, J.~H. Samson, and M.~J. Everitt, ``Simple
  procedure for phase-space measurement and entanglement validation,''
  \href{http://dx.doi.org/10.1103/PhysRevA.96.022117}{{\em Phys. Rev. A}
  {\bfseries 96}, 022117 (2017)}.

\bibitem{Hoeffding}
W.~Hoeffding, ``Probability Inequalities for Sums of Bounded Random
  Variables,'' \href{http://dx.doi.org/10.1007/978-1-4612-0865-5_26}{{\em
  Journal of the American Statistical Association} {\bfseries 58}, 13--30
  (1963)}.

\bibitem{Gurvits}
L.~Gurvits, ``On the Complexity of Mixed Discriminants and Related Problems,''
  \href{http://dx.doi.org/10.1007/11549345_39}{in {\em Mathematical Foundations
  of Computer Science 2005}, J.~J\c{e}drzejowicz and A.~Szepietowski, eds.},
  pp.~447--458.
\newblock Springer Berlin Heidelberg, Berlin, Heidelberg, 2005.

\bibitem{Lim2023}
Y.~Lim and C.~Oh, ``Approximating outcome probabilities of linear optical
  circuits,'' \href{http://dx.doi.org/10.1038/s41534-023-00791-9}{{\em npj
  Quantum Information} {\bfseries 9}, 124 (2023)}.

\bibitem{rahimi2020situ}
S.~Rahimi-Keshari, S.~Baghbanzadeh, and C.~M. Caves, ``In situ characterization
  of linear-optical networks in randomized boson sampling,''
  \href{http://dx.doi.org/10.1103/PhysRevA.101.043809}{{\em Physical Review A}
  {\bfseries 101}, 043809 (2020)}.

\bibitem{ferraro2005}
A.~Ferraro, S.~Olivares, and M.~G.~A. Paris, ``Gaussian States in Quantum
  Information,'' \href{http://arxiv.org/abs/arxiv:quant-ph/0503237}{{\ttfamily
  arxiv:quant-ph/0503237}}.

\bibitem{SchuldPRA2020}
M.~Schuld, K.~Br\'adler, R.~Israel, D.~Su, and B.~Gupt, ``Measuring the
  similarity of graphs with a Gaussian boson sampler,''
  \href{http://dx.doi.org/10.1103/PhysRevA.101.032314}{{\em Phys. Rev. A}
  {\bfseries 101}, 032314 (2020)}.

\bibitem{Lund2015}
A.~P. Lund, A.~Laing, S.~Rahimi-Keshari, T.~Rudolph, J.~L. O'Brien, and T.~C.
  Ralph, ``Boson Sampling from a Gaussian State,''
  \href{http://dx.doi.org/10.1103/PhysRevLett.113.100502}{{\em Phys. Rev.
  Lett.} {\bfseries 113}, 100502 (2014)}.

\bibitem{RahimiComplexity2015}
S.~Rahimi-Keshari, A.~P. Lund, and T.~C. Ralph, ``What Can Quantum Optics Say
  about Computational Complexity Theory?,''
  \href{http://dx.doi.org/10.1103/PhysRevLett.114.060501}{{\em Phys. Rev.
  Lett.} {\bfseries 114}, 060501 (2015)}.

\bibitem{Hamilton2017}
C.~S. Hamilton, R.~Kruse, L.~Sansoni, S.~Barkhofen, C.~Silberhorn, and I.~Jex,
  ``Gaussian Boson Sampling,''
  \href{http://dx.doi.org/10.1103/PhysRevLett.119.170501}{{\em Phys. Rev.
  Lett.} {\bfseries 119}, 170501 (2017)}.

\bibitem{lvovsky2020production}
A.~Lvovsky, P.~Grangier, A.~Ourjoumtsev, V.~Parigi, M.~Sasaki, and
  R.~Tualle-Brouri, ``Production and applications of non-Gaussian quantum
  states of light,'' \href{http://arxiv.org/abs/arXiv:2006.16985}{{\ttfamily
  arXiv:2006.16985}}.

\bibitem{hamilton2017gaussian}
C.~S. Hamilton, R.~Kruse, L.~Sansoni, S.~Barkhofen, C.~Silberhorn, and I.~Jex,
  ``Gaussian boson sampling,''
  \href{http://dx.doi.org/10.1103/PhysRevLett.119.170501}{{\em Physical review
  letters} {\bfseries 119}, 170501 (2017)}.

\bibitem{chabaud2021holomorphic}
U.~Chabaud and S.~Mehraban, ``Holomorphic representation of quantum
  computations,'' \href{http://dx.doi.org/10.22331/q-2022-10-06-831}{{\em
  Quantum} {\bfseries 6}, 831 (2022)}.

\bibitem{chabaud2023resources}
U.~Chabaud and M.~Walschaers, ``Resources for bosonic quantum computational
  advantage,'' \href{http://dx.doi.org/10.1103/PhysRevLett.130.090602}{{\em
  Physical Review Letters} {\bfseries 130}, 090602 (2023)}.

\bibitem{hertz2020quadrature}
A.~Hertz and S.~De~Bi{\`e}vre, ``Quadrature coherence scale driven fast
  decoherence of bosonic quantum field states,''
  \href{http://dx.doi.org/10.1103/PhysRevLett.124.090402}{{\em Physical Review
  Letters} {\bfseries 124}, 090402 (2020)}.

\bibitem{chabaud2020stellar}
U.~Chabaud, D.~Markham, and F.~Grosshans, ``Stellar representation of
  non-{G}aussian quantum states,''
  \href{http://dx.doi.org/10.1103/PhysRevLett.124.063605}{{\em Physical Review
  Letters} {\bfseries 124}, 063605 (2020)}.

\bibitem{henderson2024general}
L.~J. Henderson, R.~Goel, and S.~Shrapnel, ``Quantum kernel machine learning
  with continuous variables,''
  \href{http://arxiv.org/abs/arXiv:2401.05647}{{\ttfamily arXiv:2401.05647}}.

\bibitem{gao2018programmable}
Y.~Y. Gao, B.~J. Lester, Y.~Zhang, C.~Wang, S.~Rosenblum, L.~Frunzio, L.~Jiang,
  S.~Girvin, and R.~J. Schoelkopf, ``Programmable interference between two
  microwave quantum memories,''
  \href{http://dx.doi.org/10.1103/PhysRevX.8.021073}{{\em Physical Review X}
  {\bfseries 8}, 021073 (2018)}.

\bibitem{gan2020hybrid}
H.~Gan, G.~Maslennikov, K.-W. Tseng, C.~Nguyen, and D.~Matsukevich, ``Hybrid
  quantum computing with conditional beam splitter gate in trapped ion
  system,'' \href{http://dx.doi.org/10.1103/PhysRevLett.124.170502}{{\em
  Physical review letters} {\bfseries 124}, 170502 (2020)}.

\bibitem{vcernotik2024swap}
O.~{\v{C}}ernot{\'\i}k, I.~Pietik{\"a}inen, S.~Puri, S.~Girvin, and R.~Filip,
  ``Swap-test interferometry with biased qubit noise,''
  \href{http://dx.doi.org/10.1103/PhysRevResearch.6.033074}{{\em Physical
  Review Research} {\bfseries 6}, 033074 (2024)}.

\bibitem{Ising}
Z.~Wang, A.~Marandi, K.~Wen, R.~L. Byer, and Y.~Yamamoto, ``Coherent Ising
  machine based on degenerate optical parametric oscillators,''
  \href{http://dx.doi.org/10.1103/PhysRevA.88.063853}{{\em Phys. Rev. A}
  {\bfseries 88}, 063853 (2013)}.

\bibitem{dehdashti2024enhancing}
S.~Dehdashti, P.~Tiwari, K.~H.~E. Safty, P.~Bruza, and J.~Notzel, ``Enhancing
  Quantum Machine Learning: The Power of Non-Linear Optical Reproducing
  Kernels,'' \href{http://arxiv.org/abs/arXiv:2407.13809}{{\ttfamily
  arXiv:2407.13809}}.

\bibitem{frame}
C.~Ferrie and J.~Emerson, ``Frame representations of quantum mechanics and the
  necessity of negativity in quasi-probability representations,''
  \href{http://dx.doi.org/10.1088/1751-8113/41/35/352001}{{\em Journal of
  Physics A: Mathematical and Theoretical} {\bfseries 41}, 352001 (2008)}.

\bibitem{Ferrie09}
C.~Ferrie and J.~Emerson, ``Framed Hilbert space: hanging the quasi-probability
  pictures of quantum theory,''
  \href{http://dx.doi.org/10.1088/1367-2630/11/6/063040}{{\em New Journal of
  Physics} {\bfseries 11}, 063040 (2009)}.

\bibitem{weedbrook2012gaussian}
C.~Weedbrook, S.~Pirandola, R.~Garc{\'\i}a-Patr{\'o}n, N.~J. Cerf, T.~C. Ralph,
  J.~H. Shapiro, and S.~Lloyd, ``{G}aussian quantum information,''
  \href{http://dx.doi.org/10.1103/RevModPhys.84.621}{{\em Reviews of Modern
  Physics} {\bfseries 84}, 621 (2012)}.

\bibitem{boyd2004convex}
S.~Boyd, S.~P. Boyd, and L.~Vandenberghe, ``Convex optimization,''.
\newblock \href{http://dx.doi.org/10.1017/cbo9780511804441}{Cambridge
  university press}, 2004.

\bibitem{cormen2009introduction}
T.~H. Cormen, C.~E. Leiserson, R.~L. Rivest, and C.~Stein, ``Introduction to
  algorithms,''.
\newblock MIT press, 2009.

\bibitem{briegel2009measurement}
H.~J. Briegel, D.~E. Browne, W.~D{\"u}r, R.~Raussendorf, and M.~Van~den Nest,
  ``Measurement-based quantum computation,''
  \href{http://dx.doi.org/10.1038/nphys1157}{{\em Nature Physics} {\bfseries
  5}, 19--26 (2009)}.

\bibitem{knill2001scheme}
E.~Knill, R.~Laflamme, and G.~J. Milburn, ``A scheme for efficient quantum
  computation with linear optics,''
  \href{http://dx.doi.org/10.1038/35051009}{{\em nature} {\bfseries 409},
  46--52 (2001)}.

\bibitem{bartolucci2023fusion}
S.~Bartolucci, P.~Birchall, H.~Bombin, H.~Cable, C.~Dawson, M.~Gimeno-Segovia,
  E.~Johnston, K.~Kieling, N.~Nickerson, M.~Pant, {\em et al.}, ``Fusion-based
  quantum computation,''
  \href{http://dx.doi.org/10.1038/s41467-023-36493-1}{{\em Nature
  Communications} {\bfseries 14}, 912 (2023)}.

\bibitem{QMLadapBS}
U.~Chabaud, D.~Markham, and A.~Sohbi, ``Quantum machine learning with adaptive
  linear optics,'' \href{http://dx.doi.org/10.22331/q-2021-07-05-496}{{\em
  {Quantum}} {\bfseries 5}, 496 (2021)}.

\bibitem{brod2019photonic}
D.~J. Brod, E.~F. Galv{\~a}o, A.~Crespi, R.~Osellame, N.~Spagnolo, and
  F.~Sciarrino, ``Photonic implementation of boson sampling: a review,''
  \href{http://dx.doi.org/10.1117/1.AP.1.3.034001}{{\em Advanced Photonics}
  {\bfseries 1}, 034001--034001 (2019)}.

\bibitem{zhong2020quantum}
H.-S. Zhong, H.~Wang, Y.-H. Deng, M.-C. Chen, L.-C. Peng, Y.-H. Luo, J.~Qin,
  D.~Wu, X.~Ding, Y.~Hu, {\em et al.}, ``Quantum computational advantage using
  photons,'' \href{http://dx.doi.org/10.1017/cbo9780511622748.004}{{\em
  Science} {\bfseries 370}, 1460--1463 (2020)}.

\end{thebibliography}\endgroup

\widetext

\appendix


\newpage

\begin{center}
    {\huge Appendix}
\end{center}

\section{Computation of phase-space quasi-probability distributions}
\label{app:comp}

We present a general method in order to describe the \SPQD s when the states of interest $\rho(x)$ are obtained by applying a unitary circuit $U(x)$ to an initial state $\rhoin(x)$. We further assume that the initial state $\rhoin(x)= \rho_1(x)\otimes \cdots\otimes\rho_m(x)$ is a product state, in which case its \SPQD\ can be efficiently described. 
Then, the \SPQD\ of the state $\rho(x)=U(x)\rhoin(x)U(x)^\dag$ is given by
\begin{align}\label{eq:pqd-trans-func}
W^{(\bm s)}_{\rho(x)}(\bm{\alpha})=\!\int_{\mathbb C^m}\! \bdm{\beta} W^{(\bm t)}_{\rhoin(x)}(\bm{\beta})\; T^{(\bm s,-\bm t)}_{U(x)}\!(\bm\alpha|\bm\beta),
\end{align}
where $T^{(\bm s, -\bm t)}_{U(x)}\!(\bm\alpha|\bm\beta)=\pi^m \Tr[U(x)\Delta^{(-\bm t)}(\bm\alpha) U^\dagger(x) \Delta^{(\bm s)}(\bm\alpha)]$ is the transition function~\cite{rahimi2016sufficient} associated to $U(x)$, and $W^{(\bm t)}_{\rhoin(x)}(\bm{\beta})$ is the \PQD\ of the initial product state given by $W^{(\bm t)}_{\rhoin(x)}(\bm{\beta})=\prod_{j=1}^{m}W^{( t_j)}_{\rho_j(x)}\left(\beta_j\right)$. 

We note that in certain cases such as Gaussian circuits, the transition function $T^{(\bm s, -\bm t)}_{U(x)}\!(\bm\alpha|\bm\beta)$ can be efficiently computed. More generally, one can decompose the encoding circuit $U(x)=u_L(x)\cdots u_{1}(x)$ into $L$ layers of unitaries acting on at most a constant number of modes, in which case we have
\begin{equation}
\begin{aligned}\label{eq:trans-sub-trans}
T^{(\bm{r}_L, -\bm{r}_0)}_{U(x)}\!(\bm\gamma_L|\bm\gamma_0)=\!\int_{(\mathbb C^m)^{L-1}}\!\text{d}^{2m}\!\bm\gamma_1\dots\text{d}^{2m}\!\bm\gamma_{L-1}\prod_{l=1}^{L}T^{(\bm r_{l}, -\bm r_{l-1})}_{u_{l}(x)}\!(\bm\gamma_{l}|\bm\gamma_{l-1}).
\end{aligned}
\end{equation} 
Therefore, using Eqs.~\eqref{eq:pqd-trans-func} and~\eqref{eq:trans-sub-trans}, the kernel function can be expressed in terms of functions that can be computed efficiently and used to draw samples from the distribution~\eqref{eq:prob-dis}, as we show below. 

Writing the kernel function as $K(x,x')=\Tr\left[U^\dagger(x')U(x)\rho U^\dagger(x)U(x')\rho\right]$ and using  $T^{(\bm s, -\bm t)}_{U^\dagger(x)}\!(\bm\alpha|\bm\beta)=T^{(\bm s, -\bm t)}_{U(x)}\!(\bm\beta|\bm\alpha)$, one can define the probability distribution
\begin{align}\label{eq:prob-total}
P(\vec{\bm\gamma})=\frac{1}{\mathcal{N}} \big\vert W^{(\bm{r}_0)}_{\rho(x)}(\bm{\gamma}_0)\big\vert \prod_{l=1}^{L} \bigg\vert T^{( \bm{r}_{l}, -\bm{r}_{l-1})}_{u_{l}(x)}\!(\bm\gamma_{l}|\bm\gamma_{l-1})\bigg\vert\prod_{k=L+1}^{2L}\! \bigg\vert
T^{( \bm{r}_{k}, -\bm{r}_{k-1})}_{u_{k-L}(x')}\!(\bm\gamma_{k}|\bm\gamma_{k-1})\bigg\vert,
\end{align}
where $\mathcal{N}=\mathcal{N}\big(W^{(\bm{r}_0)}_{\rho(x)}\big) \prod_{l=1}^{L}\mathcal{N}\big(T^{( \bm{r}_{l}, -\bm{r}_{l-1})}_{u_{l}(x)}\big)\prod_{k=L+1}^{2L}\mathcal{N}\big(T^{( \bm{r}_{k}, -\bm{r}_{k-1})}_{u_{k-L}(x')}\big)$ is the total negative volume, and $\vec{\bm\gamma}=(\bm\gamma_0,\dots,\bm\gamma_{2L})$ is the vector of $(2L+1)m$ complex numbers. Viewing this expression as a Markov chain, by sampling form the distribution $\big\vert W^{(\bm{r}_0)}_{\rho(x)}(\bm{\gamma}_0)\big\vert/\mathcal{N}\big(W^{(\bm{r}_0)}_{\rho(x)}\big)$ associated to the initial state that is known to be product, as well as other conditional probability distributions associated to the transition functions, one can generate $N$ random samples $\vec{\bm\gamma}_1, \dots, \vec{\bm\gamma}_{N}$. Then, using the estimator
\begin{align}
	\begin{split}
E(\vec{\bm\gamma}):=\mathcal{N}_{\rho}^{(\bm{r}_0)}\text{sgn}[W^{(\bm{r}_0)}_{\rho(x')}\left(\bm{\gamma}_0\right)]\pi^m W^{(\bm{r}_{2L})}_{\rho(x')}\left(\bm{\gamma}_{2L}\right)&\prod_{l=1}^{L} \mathcal{N}_{u_l(x)}^{(\bm{r}_l)}\text{sgn}\big[T^{( \bm{r}_{l}, -\bm{r}_{l-1})}_{u_{l}(x)}\!(\bm\gamma_{l}|\bm\gamma_{l-1})\big]\\
\times&\prod_{k=L+1}^{2L}   \mathcal{N}_{u_{k-L}(x')}^{(\bm{r}_{k})} \text{sgn}\big[
T^{(\bm{r}_{k}, -\bm{r}_{k-1})}_{u_{k-L}(x')}\!(\bm\gamma_{k}|\bm\gamma_{k-1})\big],
\end{split}
\end{align}
the sample mean $\frac{1}{N}\sum_j E(\vec{\bm\gamma}_j)$ can be computed. As discussed in the main text, the relation between the estimation error and the probability of failure is given by Hoeffding's inequality. Therefore, following a similar argument, the kernel function can be estimated to within error $\epsilon=1/\text{poly}(m)$ and an exponentially small probability of failure, if one can find ordering parameters $\{\bm{r}_k\}$ such that
$\pi^m\mathcal{N}\times\big[\!\max_{\bm\gamma_{\text{b}}}\!W^{(\bm{r}_{2L})}_{\rho(x')}(\bm{\gamma}_{\text{b}})-\min_{\bm\gamma_{\text{a}}}\!W^{(\bm{r}_{2L})}_{\rho(x')}(\bm{\gamma}_{\text{a}})\!\big]$ scales polynomially with the number of modes.

Notice that, in general, this estimation algorithm is not as optimal as the estimation algorithm in terms of the \SPQD s of the data-encoding states,  see Eq.~\eqref{eq:kern-sPQD}. Indeed, the negative volume at the output of a quantum circuit is a lower bound of the product of the negative volumes of circuit elements because non-classical processes may reduce the effects of each other, e.g., a squeezing process and an anti-squeezing process can cancel each other.


\section{Decomposition of kernel estimation}
\label{app:ratio}

In this section, we show that for quantum-efficient encodings, i.e., $\Tr[\rho_{\Pi(x)}]\ge1/\mathrm{poly}(m)$ and $\Tr[\rho_{\Pi(x')}]\ge1/\mathrm{poly}(m)$, estimating the ratio $K(x,x')=\frac{\Tr[\rho_{\Pi(x)}\rho_{\Pi(x')}]}{\Tr[\rho_{\Pi(x)}]\Tr[\rho_{\Pi(x')}]}$ up to inverse-polynomial precision with exponentially small probability of failure can be done by estimating $\Tr[\rho_{\Pi(x)}]$, $\Tr[\rho_{\Pi(x')}]$ and $\Tr[\rho_{\Pi(x)}\rho_{\Pi(x')}]$ independently up to (smaller) inverse-polynomial precision with (smaller) exponentially small probability of failure, and computing the ratio of those estimates.

We rely on the following technical result:

\begin{lem}
Let $0<\epsilon<\epsilon'<1$, let $\delta>0$, let $a,b,c\in[0,1]$ and let $\tilde a,\tilde b,\tilde c\in\mathbb R$ be random variables such that
\begin{align}
|a-\tilde a|&\le\epsilon,\\
|b-\tilde b|&\le\epsilon,\\
|c-\tilde c|&\le\epsilon,
\end{align}
each with probability greater than or equal to $1-\delta$. Finally, assume that $b>\epsilon'$ and $c>\epsilon'$. Then $b,c,\tilde b,\tilde c$ do not vanish and
\begin{equation}
\left|\frac a{bc}-\frac{\tilde a}{\tilde b\tilde c}\right|\le\frac{(3+\epsilon)\epsilon}{{\epsilon'}^2(\epsilon'-\epsilon)^2},
\end{equation}
all with probability at least $1-3\delta$. 
\end{lem}

\begin{proof}
With the union bound, we have
\begin{align}
|a-\tilde a|&\le\epsilon,\\
|b-\tilde b|&\le\epsilon,\\
|c-\tilde c|&\le\epsilon,
\end{align}
all together with probability at least $1-3\delta$. In that case, $\tilde b\ge\epsilon'-\epsilon$ and $\tilde c\ge\epsilon'-\epsilon$, so $bc\tilde b\tilde c\ge{\epsilon'}^2(\epsilon'-\epsilon)^2$. This gives
\begin{align}
\left|\frac a{bc}-\frac{\tilde a}{\tilde b\tilde c}\right|&=\frac{|a\tilde b\tilde c-\tilde abc|}{bc\tilde b\tilde c}\\
&\le\frac{|a-\tilde a|bc+a|\tilde b-b|c+a\tilde b|\tilde c-c|}{\epsilon^{'2}(\epsilon'-\epsilon)^2}\\
&\le\frac{(3+\epsilon)\epsilon}{{\epsilon'}^2(\epsilon'-\epsilon)^2},
\end{align}
where we used the triangle inequality in the second line, and $a,b,c\in[0,1]$ and $\tilde b\le1+\epsilon$ in the last line.
\end{proof}

\noindent In particular, with $\epsilon=(\epsilon'/2)^4\epsilon''$ and $\delta=\delta'/3$ this implies $\left|\frac a{bc}-\frac{\tilde a}{\tilde b\tilde c}\right|\le\epsilon''$ with probability greater than or equal to $1-\delta'$. Using this lemma for $a=\Tr[\rho_{\Pi(x)}\rho_{\Pi(x')}]$, $b=\Tr[\rho_{\Pi(x)}]$ and $c=\Tr[\rho_{\Pi(x')}]$, with $\epsilon',\epsilon''=O(1/\mathrm{poly}(m))$ and $\delta=O(1/\mathrm{exp}(m))$ completes the proof.


\section{Phase-space quasi-probability distributions of states after loss channels}
\label{app:loss}

Let us consider a single-mode loss channel $\Lambda_{\eta}$ with transmissivity $0\leq\eta\leq1$. This channel reduces the amplitude of a coherent state
\begin{equation}
	\label{eq:lambda-eta}
	\Lambda_{\eta}(\ketbra{\alpha})=\ketbra{\sqrt{\eta} \alpha}.
\end{equation}
We can expand single-mode displacement operators in terms of coherent states
\begin{equation}
D(\xi)=e^{|\xi|^2/2} e^{-\xi^* a} e^{\xi a^{\dagger}}=e^{|\xi|^2/2} e^{-\xi^* a}\frac{1}{\pi}\int d^2\alpha \ketbra{\alpha} e^{\xi a^{\dagger}}=\frac{1}{\pi}e^{|\xi|^2/2}\int d^2\alpha\; e^{\xi\alpha^*-\alpha\xi^*}\ketbra{\alpha},
\end{equation}
where we used the Baker--Campbell--Hausdorff formula, the resolution of identity in terms of coherent states and  $a\ket{\alpha}=\alpha \ket{\alpha}$. By using this relation, Eq.~\eqref{eq:lambda-eta} as well as the linearity of quantum channels, we can then find the action of loss channels on single-mode displacement operators
\begin{align}
	\label{loss-dis}
	\Lambda_{\eta}(D(\xi))&=\frac{1}{\pi}e^{|\xi|^2/2}\int d^2\alpha\; e^{\xi\alpha^*-\alpha\xi^*} \Lambda_{\eta}(\ketbra{\alpha})\nonumber\\
	&=\frac{1}{\pi}e^{|\xi|^2/2}\int d^2\alpha\; e^{\xi\alpha^*-\alpha\xi^*} \ketbra{\sqrt{\eta} \alpha}\nonumber\\
	&=\frac{1}{\pi}e^{|\xi|^2/2}\int \frac{d^2\beta}{\eta}\; e^{\xi\beta^*/\sqrt{\eta}-\xi^*\beta/\sqrt{\eta}} \ketbra{\beta}\nonumber\\
	&=\frac{1}{\eta}e^{|\xi|^2/2} e^{-\xi^* a/\sqrt{\eta}} \frac{1}{\pi}\int d^2\beta \ketbra{\beta} e^{\xi a^{\dagger}/\sqrt{\eta}}\nonumber\\
	&=\frac{1}{\eta} \exp\!\left[-\Big(\frac{1}{\eta}-1\Big)\frac{|\xi|^2}{2}\right] D({\xi}/{\sqrt{\eta}}).
\end{align}
The adjoint map $\Lambda_{\eta}^*$ is related to $\Lambda_{\eta}$ through this relation
\begin{equation}
	\text{Tr}\left[ \Lambda_{\eta}(D(\xi)) D(\zeta)\right]= \text{Tr}\left[ D(\xi) \Lambda_{\eta}^*(D(\zeta))\right]\!.
\end{equation} 
Thus, using $\Tr[D(\xi)D(\eta)]=\pi\delta^{2}(\xi+\eta)$, we have
\begin{align}
	\text{Tr}\left[ \Lambda_{\eta}(D(\xi)) D(\zeta)\right]&=\frac{1}{\eta} \exp\!\left[-\Big(\frac{1}{\eta}-1\Big)\frac{|\xi|^2}{2}\right] \text{Tr}\left[ D(\xi/\sqrt{\eta}) D(\zeta)\right]\nonumber \\
	&=\frac{1}{\eta} \exp\!\left[-\Big(\frac{1}{\eta}-1\Big)\frac{|\xi|^2}{2}\right] \pi \delta^2(\xi/\sqrt{\eta}+\zeta)\nonumber\\
	&=\exp\!\left[-\Big(\frac{1}{\eta}-1\Big)\frac{|\xi|^2}{2}\right] \pi \delta^2(\xi+\zeta\sqrt{\eta})\nonumber\\
	&=\exp\!\left[-(1-\eta)\frac{|\zeta|^2}{2}\right] \pi \delta^2(\xi+\zeta\sqrt{\eta})\nonumber \\
	&=\exp\!\left[-(1-\eta)\frac{|\zeta|^2}{2}\right] \text{Tr} \left[ D(\xi) D(\zeta\sqrt{\eta}) \right]\nonumber \\
	&=\text{Tr}\left[ D(\xi) \Lambda_{\eta}^*(D(\zeta))\right].
\end{align}
This implies that the action of the adjoint map on displacement operators is given by
\begin{equation}
	\label{ed:ad-loss-dis}
	\Lambda_{\eta}^*(D(\zeta))=  \exp\!\left[-(1-\eta)\frac{|\zeta|^2}{2}\right] D(\zeta\sqrt{\eta}). 
\end{equation}
This in turn gives us the action of the adjoint map on the frame operators \eqref{eq:Delta-frame} defining the single-mode \sPQD s, 
\begin{align}\label{eq:dual-loss-on-delta}
	\Lambda_{\eta}^*(\Delta^{(s)}(\alpha))&=\int_{\mathbb C}\frac{\dm{\xi}}{\pi^{2}}\, e^{-(1-\eta)|\xi|^2/2}D(\sqrt{\eta}\xi)\, e^{s|\xi|^\dagger/2}\,e^{\alpha\xi^{*}-\xi\alpha^{*}}
	=\frac{1}{\eta} \Delta^{(s/\eta-(1-\eta)/\eta)}\big(\alpha/\sqrt{\eta}\big).
\end{align}
Employing this relation, we can express the \sPQD\ of the quantum state after loss $\rho=\Lambda_{\eta}(\rhoin)$ as
\begin{align}\label{eq:s-PQD-loss}
	W^{(\bm s)}_\rho({\alpha})=\Tr[\Lambda_{\eta}(\rhoin) \Delta^{(\bm s)}(\alpha)]
	=\Tr[\rhoin \Lambda_{\eta}^*\big(\Delta^{(\bm s)}(\alpha)\big)]
	=\frac{1}{\eta}W^{(s/\eta-(1-\eta)/\eta)}_{\rhoin}\big(\alpha/\sqrt{\eta}\big),
\end{align}
where $W^{(s)}_{\rhoin}(\alpha)$ is the \sPQD\ of the initial state $\rhoin$.

Considering the tensor product of single-mode loss channels $
\Lambda_{\bm\eta}=\Lambda_{\eta_1}\otimes\dots\otimes\Lambda_{\eta_m}$, Eq.~\eqref{eq:dual-loss-on-delta} can be generalized for multimode $\Delta^{(\bm s)}(\bm\alpha)$, defined by Eq.~\eqref{eq:Delta-frame},
\begin{align}
	\Lambda_{\bm\eta}^*\big(\Delta^{(\bm s)}(\bm\alpha)\big)&=\int_{\mathbb C^m}\frac{\bdm{\xi}}{\pi^{2m}}\, e^{-\bm\xi(I-\bm\eta)\bm\xi^\dagger/2}D(\bm\xi\sqrt{\bm\eta})\, e^{\bm\xi\bm s\xi^\dagger/2}\,e^{\bm\alpha\bm\xi^{\dagger}-\bm\xi\bm\alpha^{\dagger}}
	=\frac{1}{\det \bm\eta} \Delta^{(\bm\eta^{-1/2}(\bm{s}-I+\bm\eta)\bm\eta^{-1/2})}\big(\bm\alpha\bm\eta^{-1/2}\big),
\end{align}
where $\bm\eta=\text{diag}(\eta_1,\dots,\eta_m)$. By using this relation, the \SPQD\ of the state after an $m$-mode loss channel $\rho=\Lambda_{\bm\eta}(\rhoin)$ can be expressed in terms of the $(\bm t)$-\PQD\ of the $m$-mode initial state $\rhoin$
\begin{equation}
	W^{(\bm s)}_\rho(\bm{\alpha})=\frac{1}{\det \bm\eta} W^{(\bm t)}_{\rhoin} \big(\bm{\alpha}\bm\eta^{-1/2}\big),
\end{equation}
where $\bm t=\bm\eta^{-1/2}(\bm{s}-I+\bm\eta)\bm\eta^{-1/2}$, or equivalently $\bm s=\bm\eta^{1/2}\bm{t}\bm\eta^{1/2}+I-\bm\eta$. These reduce to the expression given in the main text in the case of diagonal matrices of ordering parameters.


\section{Estimation of lossy photonic quantum kernels using Gurvits algorithm}
\label{app:lossyGurvits}

In this section, we derive an alternative approach to classical estimation of quantum kernels based on lossy single-photon states fed into LONs based on Gurvits's algorithm for estimating the permanent~\cite{Gurvits}. 

In this case, the kernel function takes the form
\begin{align}
K(x,x')&=\Tr\Big[U(x)\bigotimes_{j=1}^m((1-\eta_j)\ketbra{0}+\eta_j\ketbra{1})U(x)^\dag U(x')\bigotimes_{j=1}^m((1-\eta_j)\ketbra{0}+\eta_j\ketbra{1})U(x')^\dag\Big]\\
&=\sum_{\bm p,\bm q\in\{0,1\}^m}\prod_{j=1}^mf_{\eta_j}(p_j)f_{\eta_j}(q_j)|\langle p_1\dots p_m|V(x,x')|q_1\dots q_m\rangle|^2,
\end{align}
where we have defined $V(x,x'):=U(x)^\dag U(x')$ and $f_\eta(p):=\eta^{1-p}(1-\eta)^p$. This means that the kernel function is equal to the expectation value of $|\langle\bm p|V(x,x')|\bm q\rangle|^2$ for $p_1,\dots,p_m$ and  $q_1,\dots,q_m$ both drawn from the product of univariate Bernoulli distributions over $\{0,1\}$ with probability $\eta_j$. Combined with Gurvits's algorithm for estimating the permanent~\cite{Gurvits}, this readily gives a classical estimation algorithm for the kernel function:
\begin{itemize}
    \item For $j\in\{1,\dots,m\}$, sample a bit  $p_j$ from the univariate Bernoulli distributions over $\{0,1\}$ with probability $\eta_j$.
    \item For $j\in\{1,\dots,m\}$, sample a bit  $q_j$ from the univariate Bernoulli distributions over $\{0,1\}$ with probability $\eta_j$.
    \item If $\|\bm p\|_1\neq\|\bm q\|_1$ output $0$ and halt. Otherwise, let $n=\|\bm p\|_1=\|\bm q\|_1$.
    \item Let $V_n(x,x')$ be the $n\times n$ matrix obtained from $V(x,x')$ by deleting the $j^{th}$ row (resp.\ $j^{th}$ column) if $p_j=0$ (resp.\ $q_j=0$) for all $j\in\{1,\dots,m\}$. 
    \item Let $W(x,x'):=V_n(x,x')\oplus V_n(x,x')^*$, such that $\mathrm{Per}[W(x,x')]=|\mathrm{Per}[V_n(x,x')]|^2$. We write $W(x,x')=(w_{ij}(x,x'))_{1\le i,j\le2n}$. 
\item Sample uniformly $N$ bit-strings $(y_1^{(l)},\dots,y_{2n}^{(l)})\in\{-1,1\}^{2n}$ for $l\in\{1,\dots,N\}$.
    \item Output $\frac1N\sum_{l=1}^Ny_1^{(l)}\cdots y_{2n}^{(l)}\prod_{i=1}^{2n}\sum_{j=1}^{2n}y_j^{(l)}w_{ij}(x,x')$.
\end{itemize}

\noindent By Hoeffding's inequality~\cite{Hoeffding} and given that the above estimator is bounded by $\|W(x,x')\|^{2n}\le1$ since $W(x,x')$ is the submatrix of a unitary matrix, the estimate obtained is an $\epsilon$-close additive estimate of the kernel function with probability at least $1-\delta$ whenever $N\ge\frac1{2\epsilon^{2}}\ln\!\left(\frac2\delta\right)$.

In case all $\eta_j$'s are equal to some $\eta\in (0,1)$, the first three steps of the above procedure can be replaced by the following ones:

\begin{itemize}
\item Compute $\theta= \sum_{n=0}^m  \binom mn^2 \eta^{2n} (1-\eta)^{2(m-n)} \in [0,1]$ and sample $b\in \{0,1\}$ from a Bernoulli distribution with parameter $\theta$, i.e., the probability of $b=1$ equals $\theta$.  

\item If $b=0$ output $0$ and halt. Otherwise, if $b=1$, sample $n\in \{0, 1, \dots, m\}$ from the binomial distribution with parameters $(m, \eta)$, i.e., the probability of picking $n$ equals $\binom mn\eta^{n} (1-\eta)^{m-n}$. 

\item Sample vectors $\bm p, \bm q \in \{0,1\}^m$ independently and uniformly at random under the constraint $\|\bm p\|_1 = \|\bm q\|_1=n$~\footnote{Here is a method for sampling a uniform $\bm p$ satisfying $\|\bm p\|_1=n$. First, let $p_1=1$ with probability $n/m$. If $p_1=1$, then let $p_2=1$ with probability $(n-1)/m-1$; if $p_1=0$, then let $p_2=1$ with probability $n/(m-1)$. Continue recursively with the rest of coordinates.}.

\end{itemize}


\section{Classical estimation of quantum kernel functions for partially measured Gaussian states}
\label{app:part-meas-G}

In this section, we give a proof of Theorem~\ref{th:part-meas-G}, which we recall below:

\setcounter{theo}{0}
\begin{theo}
    For any classical data $x$, let $\rho(x)$ be a quantum state encoding over $m$ modes obtained by performing a possibly non-Gaussian measurement of the first $k$ modes of a $(k+m)$-mode Gaussian state $\rho_G(x)$, as in Eq.~(\ref{eq:norm-state-G}). Let $\tau(x)$ denote the non-classical depth of $\rho_G(x)$ (see Eq.~(\ref{eq:s-PQD-G}) and~\cite{lee1991measure}) and let $\tau(x,x')=\max(\tau(x),\tau (x'))\in[0,\frac12]$. Then, assuming that the encoding is quantum-efficient, Algorithm~\ref{algo:MCQD-kern2} provides an estimate of the quantum kernel $K(x,x')=\Tr[\rho(x)\rho(x')]$ with additive precision $\epsilon$ and success probability $1-\delta$ in time
    \begin{equation}
        O\left(\frac{\log(\frac2\delta)\mathrm{poly}(m)}{\epsilon^2(1-\tau(x,x'))^{4k+2}}\right).
    \end{equation}
    In particular, this provides an efficient classical algorithm for quantum kernel estimation whenever $k=O(\log m)$ or $\tau(x,x')=O(\log m/k)$.
\end{theo}

\begin{proof}
Following Algorithm~\ref{algo:MCQD-kern2}, given two Gaussian states $\rho_G(x)$ and $\rho_G(x')$ over $k+m$ modes we define the state $\sigma(x,x')$ as the $(2k)$-mode Gaussian state obtained by taking the partial overlap of the last $m$ modes of $\rho_G(x)$ and $\rho_G(x')$ (see Fig.~\ref{fig:algo2}). We also denote by $\sigma(x)$ and $\sigma(x')$ the reduced states $\mathrm{Tr}_{k+1\dots k+m}[\rho_G(x)]$ and $\mathrm{Tr}_{k+1\dots k+m}[\rho_G(x')]$, respectively. 

Algorithm~\ref{algo:MCQD-kern2} combines three independent subroutines that have the same structure as Algorithm~\ref{algo:MCQD-kern1} (see Figs.~\ref{fig:algo1} and~\ref{fig:algo2}). We have shown in Appendix~\ref{app:ratio} that if each subroutine is efficient, then Algorithm~\ref{algo:MCQD-kern2} is also efficient for quantum-efficient encoding. 

We now bound the complexity of each subroutine using Eqs.~(\ref{eq:N-samples}) and (\ref{eq:range-E}). We obtain that the total number of samples for classical estimation up to additive precision $\epsilon$ and success probability $1-\delta$ is given by
\begin{equation}\label{eq:N-samples-tot}
	N\ge\frac2{\epsilon^{2}}\left[\mathcal R(E(x))^2+\mathcal R(E(x'))^2+\mathcal R(E(x,x'))^2\right]\ln\!\left(\frac2\delta\right)\!,
\end{equation}
where
\begin{align}\label{eq:range-Ex}
    \mathcal R(E(x))&=\mathcal N\big(W^{(\bm{s})}_{\sigma(x)}\big) \mathcal R(W^{(-\bm{s})}_{\Pi(x)})\\
    \mathcal R(E(x'))&=\mathcal N\big(W^{(\bm{s'})}_{\sigma(x')}\big) \mathcal R(W^{(-\bm{s'})}_{\Pi(x')})\\
    \mathcal R(E(x,x'))&=\mathcal N\big(W^{(\bm{u}\oplus-\bm v)}_{\sigma(x,x')}\big)\mathcal R(W^{(-\bm{u}\oplus\bm v)}_{\Pi(x)\otimes\Pi(x')}).
\end{align}
We choose the largest possible ordering parameters of the form $\bm s=s\bm I$ for each of the \SPQD s of the states $\sigma(x)$, $\sigma(x')$ and $\sigma(x,x')$ in all three subroutines of the algorithm, such that the corresponding \SPQD s are non-negative. Writing these parameters $s$, $s'$ and $t$, respectively, this gives
\begin{align}\label{eq:range-Ex2}
    \mathcal R(E(x))&=\mathcal R(W^{(-s\bm I)}_{\Pi(x)})\\
    \mathcal R(E(x'))&=\mathcal R(W^{(-s'\bm I)}_{\Pi(x')})\\
    \mathcal R(E(x,x'))&=\mathcal R(W^{(-t\bm I)}_{\Pi(x)\otimes\Pi(x')}),
\end{align}
where by Eq.~(\ref{eq:gen-bound-range}),
\begin{align}\label{eq:gen-bound-rangex}
    \mathcal R(W^{(-s\bm I)}_{\Pi(x)})&\le\left(\frac{2}{s+1}\right)^{k+1}\\
    \mathcal R(W^{(-s'\bm I)}_{\Pi(x')})&\le\left(\frac{2}{s'+1}\right)^{k+1}\\
    \mathcal R(W^{(-t\bm I)}_{\Pi(x)\otimes\Pi(x')})&\le\left(\frac{2}{t+1}\right)^{2k+1}\!\!\!\!\!\!\!\!\!\!.
\end{align} 
To conclude the proof, we simply need to show that the ordering parameters $s,s',t$ may be all chosen arbitrarily close to $1-2\tau(x,x')$, where $\tau(x,x')$ is the maximal non-classical depth of $\rho_G(x)$ and $\rho_G(x')$. To do so, we prove the following properties of the non-classical depth of Gaussian states, which appear to be new:

\begin{lem}\label{lem:nc-depth}
    The non-classical depth  of Gaussian states is non-increasing under partial trace and non-increasing under partial overlap.
\end{lem}

\begin{proof}
Let $\sigma$ be a Gaussian state over $k+m$ modes with covariance matrix $\Sigma$ and displacement vector $\bar{\bm r}$. Recall the expression for the \SPQD\ of a Gaussian state from Eq.~(\ref{eq:s-PQD-G}):
\begin{equation}\label{eq:s-PQD-Gapp}
W^{(\bm s)}_\sigma\left(\bm{\alpha}\right)=\frac{ e^{-\frac 12 (\bm{\alpha}- \bar{\bm r}) (\Sigma-\bm s\oplus\bm s)^{-1} (\bm{\alpha}- \bar{\bm r})^\top}}{(2\pi)^{m+k}\sqrt{\det(\Sigma-\bm s\oplus\bm s)}},
\end{equation}
for all $\bm\alpha\in\mathbb C^{k+m}$ and all $\bm s$ such that $\Sigma-\bm s\oplus\bm s$ is positive definite. By Definition~\ref{def:nc-depth}, the non-classical depth of the Gaussian state $\sigma$ is thus given by $\tau=\frac12(1-s)$, where $s$ is the supremum of the values such that $\Sigma\succ sI$, with $I$ being the $(2k+2m)\times (2k+2m)$ identity operator.

Let us write
\begin{equation}
    \Sigma=\begin{pmatrix}A&B\\B^T&C\end{pmatrix}.
\end{equation}
where $A$ is a $(2k)\times(2k)$ symmetric matrix and $C$ is a $(2m)\times(2m)$ symmetric matrix.
The condition $\Sigma\succ sI$ is equivalent to
\begin{equation}
    X^TAX+Y^TCY+2X^TB^TY>s(X^TX+Y^TY),
\end{equation}
for all $(X, Y)\in\mathbb R^{2k}\times\mathbb R^{2m}$. Setting $X=0$ gives $Y^TCY>sY^TY$ for all $Y\in\mathbb R^{2m}$, and thus $C\succ sI$, which implies that the non-classical depth of the Gaussian state obtained by taking the partial trace of $\sigma$ over the first $k$ modes (with covariance matrix $C$~\cite{weedbrook2012gaussian}) is smaller than that of $\sigma$. This shows that the non-classical depth of a Gaussian state is non-increasing under partial trace.

Let us now consider an additional Gaussian state $\sigma'$, with covariance matrix $\Sigma'$. The partial overlap over the last $m$ modes of $\sigma$ and $\sigma'$ is defined as
\begin{equation}
    \sigma'':=\Tr_{k+1,\dots,k+m,2k+m+1,\dots,2k+2m}\left[(\sigma^T\otimes\sigma')\bigotimes_{j=1}^m\ket{\mathrm{TWB}}\!\bra{\mathrm{TWB}}_{k+j,2k+m+j}\right],
\end{equation}
where $\ket{\mathrm{TWB}}:=\sum_{n\ge0}\ket{nn}$ is a infinitely-squeezed two-mode squeezed state or twin-beam state (TWB). The operator $\ket{\mathrm{TWB}}\!\bra{\mathrm{TWB}}$ can equivalently be obtained by sending a position eigenstate and a momentum eigenstate into a balanced beam-splitter~\cite{ferraro2005}, so the partial overlap can be expressed as
\begin{equation}\label{eq:partial-over-inter}
    \begin{aligned}
        \sigma''=&\Tr_{k+1,\dots,k+m,2k+m+1,\dots,2k+2m}\big[U_{BS}(\sigma^T\otimes\sigma')U_{BS}^\dag\\
        &\quad\quad\quad\times(\mathbb I_k\otimes\ket0\!\bra0_{q_{k+1}}\otimes\dots\otimes\ket0\!\bra0_{q_{k+m}}\otimes\mathbb I_k\otimes\ket0\!\bra0_{p_{2k+m+1}}\otimes\dots\otimes\ket0\!\bra0_{p_{2k+2m}})\big],
    \end{aligned}
\end{equation}
where $U_{BS}=\bigotimes_{j=1}^mU_{k+j,2k+m+j}$ is the passive linear operator corresponding to the action of the balanced beam splitters, and where $q_j$ and $p_j$ denote the position and momentum quadrature operators for the $j^{th}$ mode, respectively.

For any symmetric matrices $M$ and $M'$, $M\succ sI$ implies $OMO^T\succ sI$ for any orthogonal matrix $O$, and $M\succ sI$ and $M'\succ s'I$ implies $M\oplus M'\succ\min(s,s')I$.
The covariance matrix of the Gaussian state $\sigma''':=U_{BS}(\sigma^T\otimes\sigma')U_{BS}^\dag$ is given by $\Sigma''':=S_{U_{BS}}(T\Sigma T\oplus\Sigma')S_{U_{BS}}^T$, where $S_{U_{BS}}$ is the orthogonal matrix corresponding to the action of $U_{BS}$ on the vector of quadrature operators and $T=I\oplus(-I)$ is the orthogonal matrix corresponding to the action of the transposition on the vector of quadrature operators. In particular, its classical depth is smaller than the maximum of the non-classical depths of $\sigma$ and $\sigma'$. 

Finally, we write the covariance matrix of $\sigma'''=U_{BS}(\sigma^T\otimes\sigma')U_{BS}^\dag$ as
\begin{equation}
    \Sigma'''=\begin{pmatrix}A&B\\B^T&C\end{pmatrix},
\end{equation}
where $A$ is a $(4k)\times(4k)$ symmetric matrix and $C$ is a $(4m)\times(4m)$ symmetric matrix, and where we have ordered the vector of quadrature operators as (to get more convenient expressions later on):
\begin{equation}
    \begin{aligned}
        \bm r=(&q_1,\dots,q_k,p_1,\dots,p_k,q_{k+m+1},\dots,q_{2k+m},p_{k+m+1},\dots,p_{2k+m},\\
        &q_{k+1},\dots,q_{k+m},p_{2k+m+1},\dots,p_{2k+2m},p_{k+1},\dots,p_{k+m},q_{2k+m+1},\dots,q_{2k+2m})^T.
    \end{aligned}
\end{equation}
Then, from Eq.~(\ref{eq:partial-over-inter}) the covariance matrix of the partial overlap state $\sigma''$ is the conditional covariance matrix corresponding to a measurement of the position quadratures for the modes $k+1,\dots,k+m$ and of the momentum quadratures for the modes $2k+m+1,\dots,2k+2m$, that is~\cite{weedbrook2012gaussian}
\begin{equation}
    \Sigma''=A-B(\Pi C\Pi)^{-1}B^T,
\end{equation}
where $\Pi=I_{2m}\oplus0_{2m}$ is the projector selecting the quadratures being measured, and where the inverse is understood in the generalized (pseudo-inverse) sense, i.e., $(\Pi C\Pi)^{-1}=C_1^{-1}\Pi$, where $C_1$ is the top-left block of $C$ selected by $\Pi$. Then, the condition $\Sigma''\succ sI$ implies $C-sI\succ0$, so $C-sI$ is invertible, and
\begin{equation}
    \begin{pmatrix}A-sI&B\\B^T&C-sI\end{pmatrix}\succ0.
\end{equation}
Hence, the Schur complement of $C-sI$ in this matrix is also positive definite~\cite{boyd2004convex}, i.e.,
\begin{equation}\label{eq:schur-psd}
    A-sI-B(C-sI)^{-1}B^T\succ0.
\end{equation}
Now for all $s\ge0$ we have
\begin{equation}\label{eq:succ-s}
    (C-sI)^{-1}\succeq C^{-1}.
\end{equation}
Writing $\begin{pmatrix}C_1&C_2\\C_2^T&C_3\end{pmatrix}$, for all $(X, Y)\in\mathbb R^{2m}\times\mathbb R^{2m}$ we have
\begin{equation}\label{eq:succ-proj}
    \begin{pmatrix}X^T&Y^T\end{pmatrix}\left[C^{-1}-(\Pi C\Pi)^{-1}\right]\begin{pmatrix}
        X\\Y
    \end{pmatrix}=X^TC_1^{-1}C_2S^{-1}C_2^TC_1^{-1}X-2Y^TS^{-1}C_2^TC_1^{-1}X+Y^TS^{-1}Y,
\end{equation}
where $S=C_3-C_2^TC_1^{-1}C_2\succ0$ is the Schur complement of $C_1\succ0$, and where we have used
\begin{equation}
    (\Pi C\Pi)^{-1}=\begin{pmatrix}
        C_1^{-1}&0\\0&0
    \end{pmatrix},
\end{equation}
and the block inversion formula~\cite{cormen2009introduction}
\begin{equation}
    C^{-1}=\begin{pmatrix}
        C_1&C_2\\C_2^T&C_3
    \end{pmatrix}^{-1}=\begin{pmatrix}
        C_1^{-1}+C_1^{-1}C_2S^{-1}C_2^TC_1^{-1}&-C_1^{-1}C_2S^{-1}\\-S^{-1}C_2^TC_1^{-1}&S^{-1}
    \end{pmatrix}.
\end{equation}
Setting $X':=C_2^TC_1^{-1}X$ in Eq.~(\ref{eq:succ-proj}) we obtain
\begin{align}
    \begin{pmatrix}X^T&Y^T\end{pmatrix}\left[C^{-1}-(\Pi C\Pi)^{-1}\right]\begin{pmatrix}
        X\\Y
    \end{pmatrix}&={X'}^{T}S^{-1}X'-2Y^TS^{-1}X'+Y^TS^{-1}Y\\
    &=(X'-Y)^TS^{-1}(X'-Y)\ge0,
\end{align}
since $S^{-1}\succ0$. This implies $C^{-1}\succeq(\Pi C\Pi)^{-1}$ and together with Eq.~(\ref{eq:succ-s}) we obtain $(C-sI)^{-1}\succeq(\Pi C\Pi)^{-1}$ and thus $B(C-sI)^{-1}B^T\succeq B(\Pi C\Pi)^{-1}B^T$. With Eq.~(\ref{eq:schur-psd}), this finally yields
\begin{equation}
    \Sigma''=A-B(\Pi C\Pi)^{-1}B^T\succeq A-B(C-sI)^{-1}B^T\succ sI,
\end{equation}
when assuming that $\Sigma'''\succ sI$. This shows that the non-classical depth of $\sigma''$ is smaller than that of $\sigma'''$, which itself was smaller than the maximum of the non-classical depths of $\sigma$ and $\sigma'$. This completes the proof that the non-classical depth of Gaussian states is non-increasing under partial overlap.
\end{proof}

\noindent Since $\sigma(x)$ (resp.\ $\sigma(x')$) is obtained from $\rho_G(x)$ (resp.\ $\rho_G(x')$) by taking a partial trace, and $\sigma(x,x')$ is a partial overlap of the states $\rho_G(x)$ and $\rho_G(x')$, Lemma~\ref{lem:nc-depth} ensures that the non-classical depths of the states $\sigma(x),\sigma(x'),\sigma(x,x')$ are all bounded by $\tau(x,x')$, the maximal non-classical depth of $\rho_G(x)$ and $\rho_G(x')$. By Definition~\ref{def:nc-depth}, this implies that the ordering parameters $s,s',t$ in Eq.~(\ref{eq:gen-bound-rangex}) may be all chosen arbitrarily close to $1-2\tau(x,x')$, which concludes the proof of Theorem~\ref{th:part-meas-G}, by noting that all the covariance matrices involved can be computed in time $\mathrm{poly}(m) $.

\end{proof}


\section{Kernel estimation for adaptive Gaussian boson sampling}
\label{app:AGBS}

Given the limitations of quantum computations based on LONs for quantum kernel methods identified in the main text, we can ask whether simple extensions of LONs can restore their usefulness. One such natural extension is through the addition of \textit{adaptivity} in the measurements. Also known as feed-forward, adaptivity refers to the possibility of modifying part of the computation based on the outcomes of intermediate measurements, as in Measurement-Based Quantum Computing~\cite{briegel2009measurement} in the quantum circuit picture. Adaptivity is particularly relevant in the context of quantum computing with LONs, as it allows for performing universal quantum computations, e.g.\ through the Knill--Laflamme--Milburn scheme~\cite{knill2001scheme} or the more recent Fusion-Based Quantum Computing model~\cite{bartolucci2023fusion}. In those schemes, the addition of adaptive measurements to LONs allows for the active switching of offline resource entangled states into a LON, which can be used to implement a universal gate set on qubits encoded using photons in a near-deterministic fashion.

Kernel estimation becomes \textsf{BQP}-complete in the regime of enough adaptive measurements $k=\mathrm{poly}(m)$, where $m=\mathrm{poly}(n)$ is the number of photonic modes supporting the computations over $n$ qubits. Hence, it is expected that, unless \textsf{BPP}$=$\textsf{BQP}, estimating quantum kernels that are based on LONs with adaptive measurements is hard for classical computers (note that hardness of kernel estimation does not necessarily entails hardness of the corresponding learning task). 

What are the conditions necessary to enable quantum computational advantage through quantum kernel methods with adaptive measurements? This question was considered in~\cite{QMLadapBS} for adaptive boson sampling with input Fock states, where it was shown that if too few photons are being detected by the adaptive measurements, then there exists an efficient classical algorithm for estimating quantum kernel functions. However, since near-indistinguishable single-photon states are hard to generate experimentally in a near-deterministic fashion, a natural question is to investigate the limitations of near-term quantum computational advantages through quantum kernel methods using adaptive LONs with realistic input quantum states. In particular, we consider the case of adaptive LONs with Gaussian input states (see Fig.~\ref{fig:lossy-adap-GBS}).

This mirrors the evolution of quantum computational advantage experiments based on sampling from the output distribution of LONs, which have progressively shifted from proof-of-concept demonstration of boson sampling with input Fock states~\cite{aaronson2011computational,brod2019photonic} to Gaussian boson sampling~\cite{hamilton2017gaussian,zhong2020quantum}, where the input Fock states are replaced by Gaussian states, much easier to generate experimentally in optical platforms. 
Since single-photon states can be prepared in an offline fashion using Gaussian two-mode squeezed states and heralded photon-number measurements, and given that kernel estimation with linear optical computations using input single photons and adaptive measurements is \textsf{BQP}-complete, kernel estimation with linear optical computations using adaptive Gaussian boson sampling is also \textsf{BQP}-complete. 

In what follows, we derive sufficient conditions for efficient classical estimation of quantum kernel functions based on adaptive Gaussian boson sampling output states. In particular, we show that, similar to the case of adaptive boson sampling~\cite{QMLadapBS}, if too few photons are being detected by the adaptive measurements, then there exists an efficient classical algorithm for estimating quantum kernel functions based on adaptive Gaussian boson sampling.

\begin{figure*}[t]
	\begin{center}
		\includegraphics[width=0.9\linewidth]{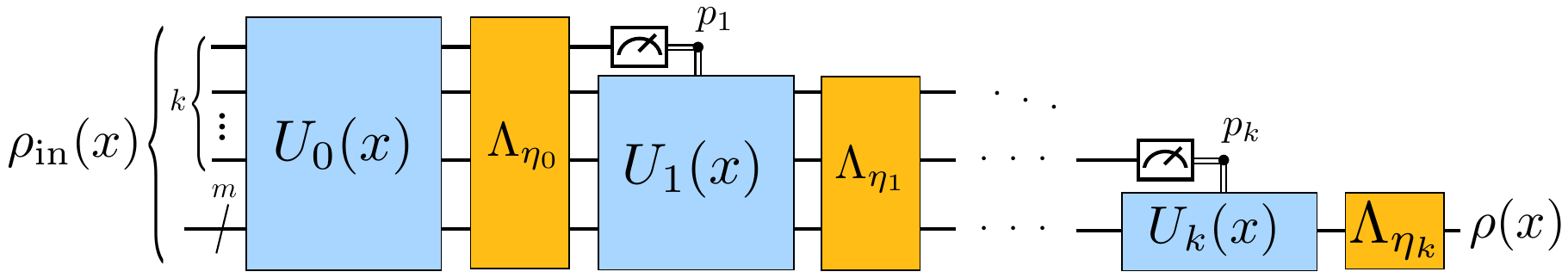}
		\caption{An adaptive Gaussian boson sampling computation, parametrised by classical data $x$. A total of $k$ modes of the initial $k+m$ modes are being adaptively measured with photon-number detection. Specific photon number measurement outcomes are being represented here, but our results cover general adaptive measurement strategies. The unitary operators $U_0(x),\dots,U_k(x)$ are Gaussian. The yellow blocks $\Lambda_{\bm\eta_0},\dots,\Lambda_{\bm\eta_k}$ represent potential layers of losses which may be non-uniform over the modes.}
		\label{fig:lossy-adap-GBS}
	\end{center}
\end{figure*}

By deferring the photon-number measurements to the end, the $m$-mode output state of a generic adaptive Gaussian boson sampling computation with $k$ adaptive measurements takes the form
\begin{equation}
    \rho=\sum_{\bm p}\mathrm{Tr}_k\!\left[(\Pi_{\bm p}\otimes\mathbb I_m)U^{(\bm p)}\rhoin U^{(\bm p)\dag}\right]\!,
\end{equation}
where the sum is over adaptive photon-number measurement patterns $\bm p=(p_1,\dots,p_k)$, with $\Pi_{\bm p}=|\bm p\rangle\langle\bm p|$, where the partial trace is over the first $k$ modes, and where
\begin{equation}
    U^{(\bm p)}:=\left(\mathbb I_k\otimes U_k^{(p_k)}\right)\!\left(\mathbb I_{k-1}\otimes U_{k-1}^{(p_{k-1})}\right)\cdots\left(\mathbb I_1\otimes U_1^{(p_1)}\right)U_0,
\end{equation}
where each $U_j$ is a Gaussian unitary over $m+k-j$ modes depending on the previous adaptive measurement outcome $p_j$ (see Fig.~\ref{fig:lossy-adap-GBS}). We restrict to Gaussian unitary operators, but a similar reasoning extends straightforwardly to Gaussian channels. 

Parametrizing the input Gaussian state $\rhoin$ and the intermediate Gaussian unitary operations with classical data $x,x'$, the corresponding quantum kernel functions take the form
\begin{equation}
    K(x,x')=\mathrm{Tr}[\rho(x)\rho(x')]=\sum_{\bm p,\bm p'}\mathrm{Tr}[\rho_{\bm p}(x)\rho_{\bm p'}(x')],
\end{equation}
where
\begin{equation}
    \rho_{\bm p}(x):=\mathrm{Tr}_k\!\left[(\Pi_{\bm p}\otimes\mathbb I_m)U^{(\bm p)}(x)\rhoin(x)U^{(\bm p)}(x)^\dag\right]\!,
\end{equation}
is a (sub-normalised) post-measurement state. The kernel thus rewrites as
\begin{equation}
    K(x,x')=\sum_{\bm p,\bm p'}K_{\bm p,\bm p'}(x,x'),
\end{equation}
with $K_{\bm p,\bm p'}(x,x'):=\mathrm{Tr}[\rho_{\bm p}(x)\rho_{\bm p'}(x')]$. Any such sub-kernel $K_{\bm p,\bm p'}(x,x')$ can be efficiently estimated through the first step of Algorithm~\ref{algo:MCQD-kern2} under the same conditions as in Theorem~\ref{th:part-meas-G}, i.e., if the number of adaptive measurements or the non-classical depth of the Gaussian states involved is small enough. This provides in turn a simple classical algorithm for estimating the full quantum kernel $K(x,x')$: writing $\mathcal S(x,x')$ a set of likely adaptive measurement patterns $\bm p,\bm p'$, i.e., which by definition satisfies
\begin{equation}
    \mathrm{Pr}[(\bm p,\bm p')\notin\mathcal S(x,x')]\le\frac1{\mathrm{poly}(m)},
\end{equation}
one may sample uniformly $(\bm p,\bm p')\in\mathcal S(x,x')$ and use Algorithm~\ref{algo:MCQD-kern2} to provide an estimate $\tilde K$ of $K_{\bm p,\bm p'}(x,x')$ up to additive precision $\epsilon$ with failure probability $\delta$. Then, by construction, 
\begin{equation}
    \big|K(x,x')-|\mathcal S(x,x')|\tilde K\big|\le\epsilon+\frac1{\mathrm{poly}(m)}.
\end{equation}
with probability $1-\delta$.

In particular, when the number of photons being detected by the adaptive measurements is too small, i.e., when $|\mathcal S(x,x')|\le\mathrm{poly}(m)$, this provides an efficient classical estimation algorithm for $K(x,x')$ under the same conditions as in Theorem~\ref{th:part-meas-G}, i.e., if the number of adaptive measurements or the non-classical depth of the Gaussian states involved is small enough. 

Note that the condition $|\mathcal S(x,x')|\le\mathrm{poly}(m)$ may be checked efficiently based on the energy of the Gaussian states $U^{(\bm p)}(x)\rhoin(x)U^{(\bm p)}(x)^\dag$ and $U^{(\bm p)}(x')\rhoin(x')U^{(\bm p)}(x')^\dag$. Moreover, $\mathcal S(x,x')$ may be chosen with an efficient classical description by picking the smallest $N(k,m,x),N'(k,m,x')$ such that $\mathcal S(N,N'):=\{(\bm p,\bm p'),\;\text{s.t.}\;|\bm p|\le N,|\bm p'|\le N'\}$ is a set of likely adaptive measurement patterns. A similar proof works for general POVM elements.

\end{document}